\documentclass[a4paper,abstracton,10pt]{article}

\usepackage{bm}
\usepackage[a4paper]{geometry}
\usepackage{booktabs}
\usepackage{stmaryrd}
\usepackage{pgfplots}
\usepackage{hyperref}
\usepackage{calc}
\usepackage{enumitem}
\usepackage{ulem}
\usepackage{mathtools}
\pgfplotsset{compat=newest}
\pgfplotsset{plot coordinates/math parser=false}

\usepackage{pifont}
\usepackage[numbers]{natbib}
\usepackage{color}
\newcommand{\xmark}{\ding{55}}%

\definecolor{darkgreen}{rgb}{0.1,0.5,0.1}
\definecolor{darkblue}{rgb}{0.1,0.1,0.9}

\newcommand{\sign}{\mathsf{sign}}

\newcommand{\E}{\mathbb{E}}
\newcommand{\trace}{\mathsf{Tr}}

\newcommand{\range}{\mathcal{R}}

\newcommand{\Span}{\mathsf{span}}
\newcommand{\MMSE}{\mathsf{MMSE}}
\newcommand{\ginvset}{\mathcal{G}}
\newcommand{\veps}{\bm{\varepsilon}}
\newcommand{\vphi}{\bm{\phi}}
\newcommand{\vpsi}{\bm{\psi}}
\newcommand{\mPhi}{\mat{\Phi}}
\newcommand{\mPsi}{\mat{\Psi}}
\newcommand{\T}{\top}
\renewcommand{\H}{{*}}

\newcommand{\vectorize}{\mathrm{vec}}

\newcommand{\lplq}[2]{\ensuremath{\ell^{#1} \rightarrow \ell^{#2}}}
\newcommand{\colpq}[2]{\ensuremath{|#1, #2|}}
\newcommand{\rowpq}[2]{\ensuremath{\overline{\underline{#1, #2}}}}
\newcommand{\ginv}[2]{\ensuremath{\mathsf{ginv}_{#1}(#2)}}
\newcommand{\pginv}[2]{\ensuremath{\mathsf{pginv}_{#1}(#2)}}

\newcommand{\ind}[1]{\mathbb{I}_{#1}}
\newcommand{\prox}{\ensuremath{\mathsf{prox}}}
\newcommand{\proj}{\ensuremath{\mathsf{proj}}}

\newcommand{\normlplq}[3]{\ensuremath{\norm{#3}_{\lplq{#1}{#2}}}}
\newcommand{\normrowpq}[3]{\ensuremath{\norm{#3}_{\rowpq{#1}{#2}}}}
\newcommand{\normcolpq}[3]{\ensuremath{\norm{#3}_{\colpq{#1}{#2}}}}

\newcommand{\singvals}[1]{\mathrm{spec}(#1)}

\usepackage{amsmath}
\usepackage{amssymb}
\usepackage{amsthm}
\usepackage{latexsym}
\usepackage{verbatim}
\usepackage{graphicx}
\usepackage{mdframed}

\definecolor{shade}{rgb}{0.93,0.93,0.93}

\newmdtheoremenv [backgroundcolor=shade, %
innertopmargin = -4pt , %
innerbottommargin =2pt , %
innerleftmargin = 1pt , %
innerrightmargin = 1pt, %
splittopskip = \topskip, %
skipbelow= 6pt, %
skipabove=6pt, %
topline=false,bottomline=false,leftline=false,rightline=false,]{theorem}{Theorem}[section]

\newmdtheoremenv[backgroundcolor=shade,%
innertopmargin = -4pt,%
innerbottommargin =2pt,%
innerleftmargin = 1pt,%
innerrightmargin = 1pt,%
splittopskip = \topskip,%
skipbelow= 6pt,%
skipabove=6pt,%
topline=false,bottomline=false,leftline=false,rightline=false,]{lemma}{Lemma}[section]

\newmdtheoremenv [backgroundcolor=shade, %
innertopmargin = -4pt , %
innerbottommargin =2pt , %
innerleftmargin = 1pt , %
innerrightmargin = 1pt, %
splittopskip = \topskip, %
skipbelow= 6pt, %
skipabove=6pt, %
topline=false,bottomline=false,leftline=false,rightline=false,]{corollary}{Corollary}[section]

\newmdtheoremenv [backgroundcolor=shade, %
innertopmargin = -4pt , %
innerbottommargin =2pt , %
innerleftmargin = 1pt , %
innerrightmargin = 1pt, %
splittopskip = \topskip, %
skipbelow= 6pt, %
skipabove=6pt, %
topline=false,bottomline=false,leftline=false,rightline=false,]{example}{Example}[section]

\newmdtheoremenv [backgroundcolor=shade, %
innertopmargin = -4pt , %
innerbottommargin =2pt , %
innerleftmargin = 1pt , %
innerrightmargin = 1pt, %
splittopskip = \topskip, %
skipbelow= 6pt, %
skipabove=6pt, %
topline=false,bottomline=false,leftline=false,rightline=false,]{definition}{Definition}[section]

\newmdtheoremenv [backgroundcolor=shade, %
innertopmargin = -4pt , %
innerbottommargin =2pt , %
innerleftmargin = 1pt , %
innerrightmargin = 1pt, %
splittopskip = \topskip, %
skipbelow= 6pt, %
skipabove=6pt, %
topline=false,bottomline=false,leftline=false,rightline=false,]{remark}{Remark}[section]

\newmdtheoremenv [backgroundcolor=shade, %
innertopmargin = -4pt , %
innerbottommargin =2pt , %
innerleftmargin = 1pt , %
innerrightmargin = 1pt, %
splittopskip = \topskip, %
skipbelow= 6pt, %
skipabove=6pt, %
topline=false,bottomline=false,leftline=false,rightline=false,]{proposition}{Proposition}[section]

{\begin{list}               
    {$\bullet$ \hfill}{
        \setlength{\leftmargin}{\parindent}
        \setlength{\parsep}{0.04\baselineskip}
        \setlength{\itemsep}{0.5\parsep}
        \setlength{\labelwidth}{\leftmargin}
        \setlength{\labelsep}{0em}}
    }
{\end{list}}

\providecommand{\cref}[1]{Chapter~\ref{chap:#1}}

\providecommand{\R}{\ensuremath{\mathbb{R}}}
\providecommand{\C}{\ensuremath{\mathbb{C}}}

\providecommand{\abs}[1]{\left|#1\right|}
\providecommand{\norm}[1]{\left\lVert#1\right\rVert}

\providecommand{\inprod}[1]{\langle#1\rangle}
\providecommand{\set}[1]{\left\lbrace#1\right\rbrace}

\providecommand{\bydef}{\overset{\mathrm{def}}{=}}
\providecommand{\diag}{\mathop{\mathrm{diag}}}
\providecommand{\rank}{\mathop{\mathrm{rank}}}

\renewcommand{\vec}[1]{\ensuremath{\mathbf{#1}}}
\providecommand{\mat}[1]{\ensuremath{\mathbf{#1}}}

\providecommand{\wh}[1]{\ensuremath{\widehat{#1}}}
\providecommand{\wt}[1]{\ensuremath{\widetilde{#1}}}

\providecommand{\mA}{\mat{A}} \providecommand{\mB}{\mat{B}}
\providecommand{\mC}{\mat{C}} \providecommand{\mD}{\mat{D}}
\providecommand{\mE}{\mat{E}} \providecommand{\mH}{\mat{H}}
\providecommand{\mF}{\mat{F}}
\providecommand{\mI}{\mat{I}} 
\providecommand{\mK}{\mat{K}}  
\providecommand{\mN}{\mat{N}}
\providecommand{\mM}{\mat{M}} \providecommand{\mP}{\mat{P}} 
\providecommand{\mQ}{\mat{Q}} \providecommand{\mR}{\mat{R}}
\providecommand{\mS}{\mat{S}} \providecommand{\mU}{\mat{U}} 
\providecommand{\mV}{\mat{V}} \providecommand{\mT}{\mat{T}}
\providecommand{\mW}{\mat{W}}
\providecommand{\mSigma}{\mat{\Sigma}}
 \providecommand{\mG}{\mat{G}}
\providecommand{\mPi}{\mat{\Pi}}
\providecommand{\mX}{\mat{X}}\providecommand{\mY}{\mat{Y}}
\providecommand{\mZ}{\mat{Z}}

\providecommand{\va}{\vec{a}} \providecommand{\vb}{\vec{b}}
\providecommand{\vc}{\vec{c}} 
\providecommand{\ve}{\vec{e}} 
 
\providecommand{\vg}{\vec{g}}
 
\providecommand{\vm}{\vec{m}} \providecommand{\vn}{\vec{n}}

\providecommand{\vu}{\vec{u}} \providecommand{\vw}{\vec{w}}
\providecommand{\vU}{\vec{U}}
\providecommand{\vx}{\vec{x}} \providecommand{\vy}{\vec{y}}
\providecommand{\vz}{\vec{z}} 
 
 \providecommand{\vv}{\vec{v}}

\providecommand{\vsigma}{\vec{\sigma}}

\DeclareMathOperator*{\argmin}{arg\,min}
\DeclareMathOperator*{\minimize}{minimize}

\renewcommand{\diag}{\mathsf{diag}}
\newcommand{\blockdiag}{\mathsf{blockdiag}}

\title{Beyond Moore-Penrose\\Part I: Generalized Inverses that Minimize Matrix Norms}
\author{Ivan Dokmani\'c and R\'emi Gribonval}
\date{}

\begin{document}

\maketitle
\normalem


\begin{abstract}
This is the first paper of a two-long series in which we study linear generalized inverses that minimize matrix norms. Such generalized inverses are famously represented by the Moore-Penrose pseudoinverse (MPP) which happens to minimize the Frobenius norm. Freeing up the degrees of freedom associated with Frobenius optimality enables us to promote other interesting properties. In this Part I, we look at the basic properties of norm-minimizing generalized inverses, especially in terms of uniqueness and relation to the MPP.

We first show that the MPP minimizes many norms beyond those unitarily invariant, thus further bolstering its role as a robust choice in many situations. We then concentrate on some norms which are generally {\em not} minimized by the MPP, but whose minimization is relevant for linear inverse problems and sparse representations. In particular, we look at mixed norms and the induced $\ell^p \rightarrow \ell^q$ norms. An interesting representative is the \emph{sparse pseudoinverse} which we study in much more detail in Part II. 

Next, we shift attention from norms to matrices with interesting behaviors. We exhibit a class whose generalized inverse is always the MPP---even for norms that normally result in different inverses---and a class for which many generalized inverses coincide, but not with the MPP. Finally, we discuss efficient computation of norm-minimizing generalized inverses.
\end{abstract}

\section{Introduction}
\label{sec:introduction}

Generalized inverses arise in applications ranging from over- and underdetermined linear inverse problems to sparse representations with redundant signal dictionaries. The most famous generalized matrix inverse is the Moore-Penrose pseudoinverse, which happens to minimize a slew of matrix norms. In this paper we study various alternatives. We start by introducing some of the main applications and discussing motivations.

\paragraph{Linear inverse problems.} In discrete linear inverse problems, we seek to estimate a signal $\vx$ from measurements $\vy$, when they are related by a linear system, $\vy = \mA \vx + \vn$, $ \mA \in \C^{m \times n}$ up to a noise term $\vn$. Such problems come in two rather different flavors: overdetermined ($m>n$) and underdetermined ($m<n$). Both cases may occur in the same application, depending on how we tune the modeling parameters. For example, in computed tomography the entries of the system matrix quantify how the $i$th ray affects the $j$th voxel. If we target a coarse resolution (fewer voxels than rays), $\mA$ is tall and we deal with an overdetermined system. In this case, we may estimate $\vx$ from $\vy$ by applying a generalized (left) inverse, very often the Moore-Penrose pseudoinverse (MPP).\footnote{The Moore-Penrose pseudoinverse was discovered by Moore in 1920 \cite{Moore:1920}, and later independently by Penrose in 1955 \cite{Penrose:2008foa}.}
When the system is underdetermined ($m<n$), we need a suitable signal model to get a meaningful solution. As we discuss in Section \ref{sub:MPP}, for most common models (e.g. sparsity), the reconstruction of $\vx$ from $\vy$ in the underdetermined case is no longer achievable by a linear operator in the style of the MPP.

\paragraph{Redundant representations.} In redundant representations, we represent lower-dimensional vectors through higher-dimensional frame and dictionary expansions. The frame expansion coefficients are computed as $\bm{\alpha} = \mA^\H \vx$, where the columns of a fat $\mA$ represent the frame vectors, and $\mA^{\H}$ denotes its conjugate transpose. The original signal is then reconstructed as $\vx = \mD \bm{\alpha}$, where $\mD$ is a \emph{dual frame} of $\mA$, such that $\mD\mA^\H = \mI$. There is a unique correspondence between dual frames and generalized inverses of full rank matrices. Different duals lead to different reconstruction properties in terms of resilience to noise, resilience to erasures, computational complexity, and other figures of merit \cite{Kovacevic:2008tw}. It is therefore interesting to study various duals, in particular those optimal  according to various criteria; equivalently, it is interesting to study various generalized inverses.

\paragraph{Generalized inverses beyond the MPP.} In general, for a matrix $\mA \in \C^{m \times n}$, there are many different generalized inverses. If $\mA$ is invertible, they all match. The MPP $\mA^{\dagger}$ is special because it optimizes a number of interesting properties. Much of this optimality comes from geometry: for $m<n$, $\mA^{\dagger} \mA$ is an orthogonal projection onto the range of $\mA^\H$, and this fact turns out to play a key role over and over again. Nevertheless, the MPP is only one of infinitely many generalized inverses, and it is interesting to investigate the properties of others. As the MPP minimizes a particular matrix norm---the Frobenius norm\footnote{We will see later that it actually minimizes many norms.}---it seems natural to study alternative generalized inverses that minimize different matrix norms, leading to different optimality properties. Our initial motivation for studying alternative generalized inverses is twofold:

\begin{enumerate}[label=(\roman*)]
    \item \emph{Efficient computation:} Applying a sparse pseudoinverse
    requires less computation than applying a full one
    \cite{Dokmanic:2013bo,Li:2013cx,Casazza:ev,Krahmer:2012cy}.
    We could take advantage of this fact if we knew how to compute a sparse pseudoinverse that is in some sense stable to noise. This sparsest pseudoinverse may be formulated as
    \begin{equation}
        \label{eq:intro_l0}
        \ginv{0}{\mA} \bydef \argmin_{\mX}~\norm{\vectorize(\mX)}_{0}\ \text{subject to} \ \mA\mX = \mI
    \end{equation}
    where $\norm{\, \cdot \,}_{0}$ counts the total number of non-zero entries
    in a vector and $\vectorize(\cdot)$ transforms a matrix into a vector by stacking its columns. The non-zero count gives the naive complexity of applying
    $\mX$ or its adjoint to a vector. Solving the optimization problem \eqref{eq:intro_l0} is in general NP-hard \cite{natarajan95:_spars,DMA97}, although we will see that for most matrices $\mA$ finding \emph{a} solution is trivial and not very useful: just invert any full-rank $m \times m$ submatrix and zero the rest. This strategy is not useful in the sense that the resulting matrix is poorly conditioned. On the other hand, the vast literature establishing equivalence between
    $\ell^{0}$ and $\ell^{1}$ minimization suggests to replace \eqref{eq:intro_l0} by the
    minimization of the entrywise $\ell^1$ norm
    \newcommand{\vkl}{\mathbf{a}}
    \begin{equation}
      \label{eq:l1_instead_l0}
      \ginv{1}{\mA} \bydef \argmin_{\mX}~\norm{\vectorize(\mX)}_1\ \text{subject to} \ \mA \mX = \mI.
    \end{equation}
    Not only is \eqref{eq:l1_instead_l0} computationally tractable, but we will show in Part II that unlike inverting a submatrix, it also leads to well-conditioned matrices which are indeed sparse.

    \item \emph{Poor man's $\ell^{p}$ minimization:}
    Further motivation for alternative generalized inverses comes from an idea to construct a linear poor man's version of the $\ell^p$-minimal solution to an underdetermined set of linear equations $\vy = \mA \vx$. For a general $p$, the solution to
    \begin{equation}
      \wh{\vx} \bydef \argmin \norm{\vx}_p \ \text{subject to} \ \mA \vx = \vy,
    \end{equation}
    cannot be obtained by any linear operator $\mB$ (see Section \ref{sub:MPP}). That is, there is no $\mB$ such that $\vz \bydef \mB \vy$  satisfies $\wh{\vx} = \vz$ for every choice of $\vy$. The exception is $p=2$ for which the MPP does provide the minimum $\ell^{2}$ norm representation $\mA^{\dagger} \vy$ of $\vy$; Proposition \ref{prop:only_MPP_is_linear} and comments thereafter show that this is indeed the only exception. On the other hand, we can obtain the following bound, valid for any $\vx$ such that $\mA\vx = \vy$, and in particular for $\wh{\vx}$: 
    \begin{equation}
    \norm{\vz}_{p} =  \norm{\mB \mA \vx}_p \leq \normlplq{p}{p}{\mB \mA} \norm{\vx}_p,
    \end{equation}
    where $\normlplq{p}{p}{\, \cdot \,}$ is the operator norm on matrices induced by the $\ell^{p}$ norm on vectors.
    If $\mA\mB = \mI$, then $\vz = \mB \vy$ provides an admissible representation $\mA \vz = \vy$, and
    \begin{equation}
      \norm{\vz}_p \leq \normlplq{p}{p}{\mB \mA} \norm{\wh{\vx}}_p.
    \end{equation}
    This expression suggests that the best linear generalized inverse $\mB$ in the sense of
    minimal worst case $\ell^p$ norm blow-up is the one that minimizes $\normlplq{p}{p}{\mB
      \mA}$, motivating the definition of
      \begin{equation}
      \pginv{\lplq{p}{p}}{\mA} \bydef \argmin_{\mX}~\norm{\mX\mA}_{\lplq{p}{p}}\ \text{subject to} \ \mA \mX = \mI.
    \end{equation}
\end{enumerate}

\paragraph{Objectives.}
Both (i) and (ii) above are achieved by minimization of some matrix norm. The purpose of this paper is to investigate the properties of generalized inverses $\ginv{}{\cdot}$ and $\pginv{}{\cdot}$ defined using various norms by addressing the following questions:
\begin{enumerate}
\item Are there norm families that all lead to the same generalized inverse,
  thus facilitating computation?
\item Are there specific classes of matrices for which different norms lead
  to the same generalized inverse, potentially different from the MPP?
\item Can we quantify the stability of matrices that result from these optimizations? In particular, can we control the Frobenius norm of the sparse pseudoinverse $\ginv{1}{\mA}$, and more generally of any $\ginv{p}{\mA}$, $p \geq 1$, for some random class of $\mA$? This is the topic of Part II.
\end{enumerate}

\subsection{Prior Art}

Several recent papers in frame theory study alternative dual frames, or equivalently, generalized inverses.\footnote{Generalized inverses of full-rank matrices.} These works concentrate on existence results and explicit constructions of sparse frames and sparse dual frames with prescribed spectra \cite{Casazza:ev, Krahmer:2012cy}. Krahmer, Kutyniok, and Lemvig \cite{Krahmer:2012cy} establish sharp bounds on the sparsity of dual frames, showing that generically, for $\mA \in \C^{m \times n}$, the sparsest dual has $mn - m^2$ zeros. Li, Liu, and Mi \cite{Li:2013cx} provide bounds on the sparsity of duals of Gabor frames which are better than generic bounds. They also introduce the idea of using $\ell^p$ minimization to compute these dual frames, and they show that under certain circumstances, the $\ell^p$ minimization yields the sparsest possible dual Gabor frame. Further examples of non-canonical dual Gabor frames are given by Perraudin \emph{et al.}, who use convex optimization to derive dual Gabor frames with more favorable properties than the canonical one \cite{Perraudin:2014we}, particularly in terms of time-frequency localization.

Another use of generalized inverses other than the MPP is when we have some idea about the subspace we want the solution of the original inverse problem to live in. We can then apply the restricted inverse of Bott and Duffin \cite{Bott:1953ur}, or its generalizations \cite{Minamide:1970ug}. The authors in \cite{Suarez:2010jo} show how to compute \emph{approximate} MPP-like inverses with an additional constraint that the minimizer lives in a particular matrix subspace, and how to use such matrices to precondition linear systems.

An important use of frames is in channel coding where the need for robust reconstruction in presence of channel errors leads to the design of optimal dual frames \cite{JinsongLeng:ck, Lopez:2010ky}. Similarly, one can try to compute the best generalized inverse for reconstruction from quantized measurements \cite{Lammers:2008dk}. This is related to the concept of Sobolev dual frames which minimize matrix versions of $\ell^2$-type Sobolev norms \cite{Blum:2009cm} and which admit closed-form solutions.\footnote{Our \emph{poor man's $\ell^p$ minimization} in Section \ref{sec:normmin} has a similar formulation but without a closed-form solution.} The authors in \cite{Blum:2009cm} show that these alternative duals give a linear reconstruction scheme for $\Sigma \Delta$ quantization with a lower asymptotic reconstruction error than the cannonical dual frame (the MPP). Sobolev dual frames have also been used in compressed sensing with quantized measurements \cite{Gunturk:2012go}. Some related ideas go back to the Wexler-Raz identity and its role in norm-minimizing dual functions \cite{Daubechies:1994ev}.

A major role in the theory of generalized inverses and matrix norms is played by unitarily invariant norms, studied in depth by Mirsky \cite{Mirsky:1960fi}. Many results on the connection between these norms and the MPP are given by Zi\c{e}tak \cite{Zietak:1997ut}; we comment on these connections in detail in Section \ref{sec:normsmpp}. In their detailed account of generalized inverses \cite{BenIsrael:2003fi}, Ben-Israel and Greville use the expression \emph{minimal properties of generalized inverses}, but they primarily concentrate on variations of the square-norm minimality. Additionally, they define a class of non-linear generalized inverses corresponding to various metric projections. We are primarily concerned with generalized inverses that are themselves matrices, but one can imagine various \emph{decoding rules} that search for a vector satisfying a model, and being consistent with the measurements \cite{Bourrier:2014dg}. In general, such decoding rules are not linear.

Finally, sparse pseudoinverse was previously studied in \cite{Dokmanic:2013bo}, where it was shown empirically that the minimizer is indeed a sparse matrix, and that it can be used to speed up the resolution of certain inverse problems.

\subsection{Our Contributions and Paper Outline}

We study the properties of generalized inverses corresponding to norms\footnote{Strictly speaking, we also consider {\em quasi}-norms (typically for $0< p,q <1$)} listed in Table~\ref{tab:norm-list}, placed either on the candidate inverse $\mX$ itself ($\ginv{}{\cdot}$) or on the projection $\mX \mA$ ($\pginv{}{\cdot}$). 

A number of relevant definitions and theoretical results are laid out in Section~\ref{sec:theory}. In Section \ref{sec:normmin} we put forward some preliminary results on norm equivalences with respect to norm-minimizing generalized inverses.
We also talk about poor man's $\ell^p$ minimization, by discussing generalized inverses that minimize the worst case and the average case $\ell^p$ blowup. These inverses generally do not coincide with the MPP. We then consider a property of some random matrix ensembles with respect to norms that do not lead to the MPP, and show that they satisfy what we call the {\em unbiasedness property}.

Section \ref{sec:normsmpp} discusses classes of norms that lead to the MPP. We extend the results of Zi\c{e}tak on unitarily invariant norms to left-unitarily invariant norms which is relevant when minimizing the norm of the projection operator $\mX \mA$, which is in turn relevant for {\em poor man's $\ell^p$ minimization} (Section \ref{sub:poor_mans_lp_minimization}). We conclude Section \ref{sec:normsmpp} by a discussion of norms that almost never yield the MPP. We prove that most mixed norms (column-wise and row-wise) almost never lead to the MPP.  A particular representative of these norms is the entrywise $\ell^1$ norm giving the \emph{sparse pseudoinverse}. Two fundamental questions about sparse pseudoinverses---those of uniqueness and stability---are discussed in Part II. 

While Section \ref{sec:normsmpp} discusses norms, in Section \ref{sec:same_inverse_for_many_norms} we concentrate on matrices. We have seen that many norms yield the MPP for all possible input matrices and that some norms generically do not yield the MPP. In Section \ref{sec:same_inverse_for_many_norms} we first discuss a class of matrices for which some of those latter norms in fact \emph{do} yield the MPP. It turns out that for certain matrices whose MPP has ``flat'' columns (\emph{cf.} Theorem \ref{th:TFUC}), $\ginv{\nu}{\mA}$ contains the MPP for a large class of mixed norms $\nu$, also those that normally do not yield the MPP. This holds in particular for partial Fourier and Hadamard matrices.  Next, we exhibit a class of matrices for which many generalized inverses coincide, but \emph{not} with the MPP.

Finally, in Section \ref{sec:computation} we discuss how to efficiently compute many of the mentioned pseudoinverses. We observe that in some cases the computation simplifies to a vector problem, while in other cases it is indeed a full matrix problem. We use the alternating-direction method of multipliers (ADMM) \cite{Parikh:2013} to compute the generalized inverse, as it can conveniently address both the norms on $\mX$ and on $\mX \mA$.

\subsection{Summary and Visualization of Matrix Norms}

We conclude the introduction by summarizing some of the results and norms in Table \ref{tab:norm-list} and using the \emph{matrix norm cube} in Figure \ref{fig:norm-cube}---a visualization gadget we came up with for this paper. We place special emphasis on MPP-related results.

The norm cube is an effort to capture the various equivalences between matrix norms for particular choices of parameters. Each point on the displayed planes corresponds to a matrix norm. More details on these equivalences are given in Section \ref{sec:normmin}. For example, using notations in Table \ref{tab:norm-list}, the Schatten $2$-norm equals the Frobenius norm as well as the entrywise $2$-norm. The induced $\lplq{2}{2}$ norm equals the Schatten $\infty$-norm, while the induced $\lplq{1}{p}$ norm equals the largest column $p$-norm, that is to say the $\colpq{p}{\infty}$ columnwise mixed norm. 

For many matrix norms $\nu$, we prove (\emph{cf.} Corollary \ref{cor:classicalMPP}) that $\ginv{\nu}{\mA}$ and $\pginv{\nu}{\mA}$ always contain the MPP. For other norms (\emph{cf.} Theorem~\ref{thm:not_always_mpp}) we prove the existence of matrices $\mA$ such that $\ginv{\nu}{\mA}$ (resp. $\pginv{\nu}{\mA}$) does not contain the MPP. This is the case for poor man's $\ell^{1}$ minimization, both in its worst case flavor $\pginv{\lplq{1}{1}}{\mA}=\pginv{\colpq{1}{\infty}}{\mA}$ and in an average case version $\pginv{\rowpq{2}{1}}{\mA}$, \emph{cf.} Section~\ref{sub:poor_mans_lp_minimization} and Proposition~\ref{prop:avg_case_minimization}.

A number of questions remain open. Perhaps the main group is to induced norms for $q_{\text{o}} \neq 2$ and their relation to the MPP.

\begin{figure}[h!]
  \centering
  \includegraphics[width=\textwidth]{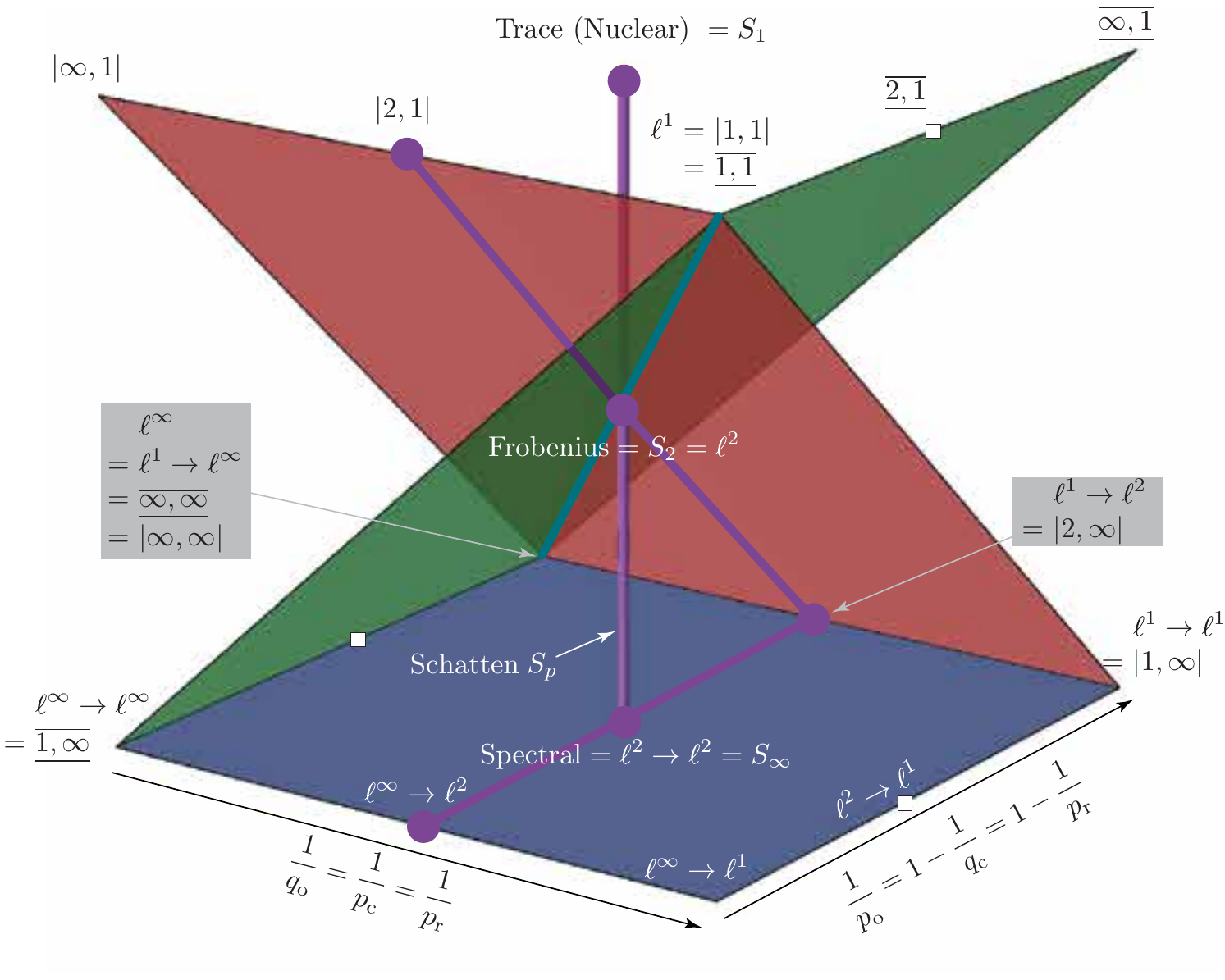}
  \caption{
  The blue plane is that of operator norms $\lplq{p_{\textnormal{o}}}{q_{\textnormal{o}}}$, the red one of columnwise mixed norms $\colpq{p_{\textnormal{c}}}{q_{\textnormal{c}}}$, and the green one of rowwise mixed norms $\rowpq{p_{\textnormal{r}}}{q_{\textnormal{r}}}$. The intersection of rowwise and columnwise mixed norms is shown by the thick blue line---these are entrywise $\ell^{p}$ norms. The vertical gray line is that of Schatten norms, with the nuclear norm $S_{1}$ at the top, the Frobenius norm $S_{2}$ in the middle (intersecting entrywise norms), and the spectral norm $S_\infty$ at the bottom (intersecting operator norms). Among columnwise mixed norms $\colpq{p_{\textnormal{c}}}{q_{\textnormal{c}}}$, all norms with a fixed value of $p_{\textnormal{c}}$ lead to the same minimizer. Purple circles and lines indicate norms $\nu$ for which $\ginv{\nu}{\mA}$ and $\pginv{\nu}{\mA}$ contain the MPP.}
  \label{fig:norm-cube}
\end{figure}

\begin{table}[h!]
  \centering
  \scalebox{0.8}{
  \begin{tabular}{@{}lll|l|l@{}}
    \toprule
    {\bf Norm name} & {\bf Symbol} & {\bf Definition} & $\mA^{\dagger} \in \ginv{\nu}{\mA}$ & $\mA^{\dagger} \in \pginv{\nu}{\mA}$ \hfill [Proof] \\
    \midrule
    Schatten
  & $\norm{\mM}_{S_p}$
  & $\norm{\singvals{\mM}}_{p}$ 
  & \checkmark $1 \leq p \leq \infty$ 
  & \checkmark $1 \leq p \leq \infty$ \hfill [Cor~\ref{cor:classicalMPP}-1]  \\
\hline
  Columnwise
  & $\normcolpq{p}{q}{\mM}$
  & $\norm{\{\norm{\vm_j}_p\}_{j=1}^n}_{q}$ 
  & \checkmark $p=2$, $1\leq q \leq \infty$ 
  & \checkmark $p=2$, $1\leq q \leq \infty$ \hfill [Cor~\ref{cor:classicalMPP}-2]\\ 
  mixed norm&&
   & \xmark \ $p \neq 2$, $0<p,q<\infty$ 
   & \xmark \ $p \neq 2$, $1\leq p<\infty$, $0<q <\infty$ \hfill [Thm~\ref{thm:not_always_mpp}-1\&3] \\
  \hline
  Entrywise
  & $\norm{\mM}_p$
  & $\norm{\vectorize(\mM)}_p$ 
  & \checkmark $p=2$ 
  & \checkmark $p=2$ \hfill [Cor~\ref{cor:classicalMPP}-2]\\ 
   &&
   & \xmark \ $p \neq 2$, $0<p<\infty$ 
   & \xmark \ $p \neq 2$, $0<p<\infty$ \hfill [Thm~\ref{thm:not_always_mpp}-1\&3] \\
  \hline
  Rowwise
  & $\normrowpq{p}{q}{\mM}$
  & $\norm{\{\norm{\vm^i}_p\}_{i=1}^m}_{q}$  
  & \checkmark $(p,q)=(2,2)$ 
  & \checkmark $(p, q)= (2, 2)$\hfill [Cor~\ref{cor:classicalMPP}-2] \\
 mixed norm &&
 &  \xmark\ $(p,q) \neq (2,2)$, $0<p,q<\infty$ 
  & \xmark\ $(p,q) \neq (2,2)$, $1 \leq p < \infty$, $0<q<\infty$  \hfill [Thm~\ref{thm:not_always_mpp}-2\&4]\\
\hline
  Induced   & $\normlplq{p}{q}{\mM}$
  & $\sup_{\vx \neq 0} \frac{\norm{\mM \vx}_q}{\norm{\vx}_p}$ 
  & \checkmark $1 \leq p \leq \infty$, $q=2$   
  & \checkmark $1 \leq p \leq \infty$, $q=2$ \hfill [Cor~\ref{cor:classicalMPP}-3]   \\
    \bottomrule
     \end{tabular}}

  \caption{Summary of matrix norms considered in this paper. 
  Notations are mostly introduced in Section~\ref{sec:theory}. 
  Checkmarks (\checkmark) indicate parameter values for which the property is true for any full-rank \emph{real or complex} matrix $\mA$ 
  Crosses (\xmark) indicate parameter values for which for any $m<n$, $m \geq 3$, there exists a full-rank $\mA \in \R^{m\times n}$ such that $\mA^{\dag} \notin \ginv{\nu}{\mA}$ (resp. $\mA^{\dag} \notin \pginv{\nu}{\mA}$). 
  }

  \label{tab:norm-list}
\end{table}

\section{Definitions and Known Results}
\label{sec:theory}

Throughout the paper we assume that all vectors and matrices are over $\C$ and point out when a result is valid only over the reals. Vectors are all column vectors, and they are denoted by bold lowercase letters, like $\vx$. Matrices are denoted by bold uppercase letters, such as $\mM$. By $\mM \in \C^{m \times n}$ we mean that the matrix $\mM$ has $m$ rows and $n$ columns of complex entries. The notation $\mI_{m}$ stands for the identity matrix in $\C^{m \times m}$; the subscript $m$ will often be omitted. We write $\vm_{j}$ for the $j$th column of $\mM$, and $\vm^{i}$ for its $i$th row. The conjugate transpose of $\mM$ is denoted $\mM^\H$ and the transpose is denoted $\mM^{\T}$. The notation $\ve_{i}$ denotes the $i$th canonical basis vector. Inner products are denoted by $\inprod{\cdot, \cdot}$. All inner product are over complex spaces, unless otherwise is indicated in the subscript. For example, an inner product over real $m \times n$ matrices will be written $\inprod{\cdot,\, \cdot}_{\R^{m\times n}}$

\subsection{Generalized Inverses}

A generalized inverse of a rectangular matrix is a matrix that has some, but not all properties of the standard inverse of an invertible square matrix. It can be defined for non-square matrices that are not necessarily of full rank.

\begin{definition}[Generalized inverse]
  \label{def:generalized_inverse}
  $\mX \in \C^{n \times m}$ is a generalized inverse of a matrix $\mA \in \C^{m
    \times n}$ if it satisfies $\mA \mX \mA = \mA$.
\end{definition}
We denote by $\ginvset(\mA)$ the set of all generalized inverses of a matrix $\mA$ .

For the sake of clarity we will primarily concentrate on inverses of underdetermined matrices ($m<n$). As we show in Section \ref{sub:fat_and_skinny}, this choice does not incur a loss of generality. Furthermore, we will assume that the matrix has full rank: $\rank(\mA) = m$. In this case, $\mX$ is the generalized (right) inverse of $\mA$ if and only if $\mA \mX = \mI_{m}$. 

\subsection{Correspondence Between Generalized Inverses and Dual Frames}

\begin{definition}
A collection of vectors $(\vphi_i)_{i=1}^n$ is called a (finite) \emph{frame} for $\C^m$
if there exist constants $A$ and $B$, $0 < A \leq B <
\infty$, such that

\begin{equation}
  A \norm{\vx}_2^2 \leq \sum_{i=1}^n \abs{\inprod{\vx, \vphi_i}}^{2} \leq B \norm{\vx}_2^2,
\end{equation}
for all $\vx \in \C^m$.
\end{definition}

\begin{definition}
A frame $(\vpsi_i)_{i=1}^n$ is a \emph{dual frame} to $(\vphi_i)_{i=1}^n$ if
the following holds for any $\vx \in \C^m$,

\begin{equation}
  \vx = \sum_{i = 1}^n \inprod{\vx, \vphi_i} \vpsi_i.
\end{equation}
\end{definition}
This can be rewritten in matrix form as 

\begin{equation}
  \vx = \mPsi \mPhi^\H \vx.
\end{equation}
As this must hold for all $\vx$, we can conclude that

\begin{equation}
  \mPsi \mPhi^\H = \mI_m,
\end{equation}
and so any dual frame $\mPsi$ of $\mPhi$ is a generalized left inverse of $\mPhi^\H$. Thus there is a one-to-one correspondence between dual frames and generalized inverses of full rank matrices.

\subsection{Characterization with the Singular Value Decomposition (SVD)}

A particularly useful characterization of generalized inverses is through the singular value decomposition (SVD). This characterization has been used extensively to prove theorems in \cite{Krahmer:2012cy,Zietak:1997ut} and elsewhere. Consider the SVD of the matrix $\mA$
\begin{equation}
 \label{eq:DefSVD}
  \mA = \mU \mSigma \mV^\H,
\end{equation}
where $\mU \in \C^{m \times m}$ and $\mV \in \C^{n \times n}$ are unitary and $\mSigma = \left[\diag(\sigma_1(\mA), \ldots, \sigma_m(\mA)),\mathbf{0}_{m \times (n-m)}\right]$  contains the singular values of $\mA$ in a non-increasing order. For a matrix $\mX$, let $\mM \bydef \mV^\H \mX \mU$. Then it follows from Definition \ref{def:generalized_inverse} that $\mX$ is a generalized inverse of $\mA$ if and only if
\begin{equation}
  \mSigma \mM \mSigma = \mSigma.
\end{equation}
Denoting by $r$ the rank of $\mA$ and setting
\begin{equation}
\label{eq:DefSingValues}
\mSigma_{\square} = \diag(\sigma_1(\mA), \ldots, \sigma_r(\mA)),
\end{equation}
we deduce that $\mM$ must be of the form
\begin{equation}
    \label{eq:Mform_rankdeff}
    \mM = 
    \begin{bmatrix}
        \mSigma_{\square}^{-1} & \mR \\
        \mS & \mT
    \end{bmatrix}
\end{equation}
where  $\mR \in \C^{r \times (m - r)}$, $\mS \in
\C^{(n-r) \times r}$, and $\mT \in \C^{(n-r) \times (m-r)}$ are arbitrary
matrices.

For a full-rank $\mA$, \eqref{eq:Mform_rankdeff} simplifies to 
\begin{equation}
  \label{eq:Mform}
  \mM = \left[ \begin{array}{c}\mSigma_{\square}^{-1}\\  \mS\end{array} \right],
\end{equation}
In the rest of the paper we restrict ourselves to full-rank matrices, and use the following characterization of the set of all generalized inverses of a matrix $\mA = \mU \mSigma \mV^\H$:
\begin{equation}
\label{eq:GinvsetMform}
  \ginvset(\mA) = \set{\mX \ : \ \mX = \mV \mM \mU^\H \ \mbox{where}\ \mM\
  \mbox{has the form}\ \eqref{eq:Mform}}.
\end{equation}
For rank-deficient matrices, the same holds with~\eqref{eq:Mform_rankdeff} instead of~\eqref{eq:Mform}. Using this alternative characterization to extend the main results of this paper to rank-deficient matrices is left to future work.

\subsection{The Moore-Penrose Pseudoinverse (MPP)}
\label{sub:MPP}

The Moore-Penrose Pseudoinverse (MPP) has a special place among generalized inverses, thanks to its various optimality and symmetry properties.

\begin{definition}[MPP]
  \label{def:MPP}
  The Moore-Penrose pseudoinverse of the matrix $\mA$ is the unique matrix
  $\mA^{\dagger}$ such that
  \begin{align}
    &\mA \mA^{\dagger} \mA = \mA,
    &(\mA \mA^{\dagger})^\H = \mA \mA^{\dagger}, \nonumber \\
    &\mA^{\dagger} \mA \mA^{\dagger} = \mA^{\dagger}, &(\mA^{\dagger}
    \mA)^\H = \mA^{\dagger} \mA.
  \end{align}
\end{definition}

This definition is universal---it holds regardless of whether $\mA$ is underdetermined or overdetermined, and regardless of whether it is full rank.

Under the conditions primarily considered in this paper ($m<n$, $\operatorname{rank}(\mA) = m$), we can express the MPP as $\mA^{\dagger} = \mA^\H (\mA \mA^\H)^{-1}$, which corresponds to the particular choice $\mS = \mathbf{0}_{(n-m)\times m}$ in~\eqref{eq:Mform}. The canonical dual frame $\mPhi$ of a frame $\mPsi$ is the adjoint $\mPhi = [\mPhi^{\dagger}]^\H$ of its MPP.

There are several alternative definitions of MPP. One that pertains to our
work is:
\begin{definition}
  MPP is the unique generalized inverse of $\mA$ with minimal Frobenius norm.
\end{definition}
That this definition makes sense will be clear from the next section. As
we will see in Section \ref{sec:normsmpp}, the MPP can also be characterized
as the generalized inverse minimizing many other matrix norms.

MPP has a number of interesting properties. If $\mA \in \C^{m \times n}$, with
$m > n$, and
\begin{equation}
  \label{eq:overdetermined_noisy}
  \vy = \mA \vx + \ve,
\end{equation}
we can compute
\begin{equation}
  \wh{\vx} = \mA^{\dagger} \vy.
\end{equation}
This vector $\wh{\vx}$ is what would in the noiseless case generate $\wh{\vy} = \mA \mA^{\dagger} \vy$---the orthogonal projection of $\vy$ onto the range of $\mA$, $\range(\mA)$. This is also known as the least-squares solution to an inconsistent overdetermined system of linear equations, in the sense that it minimizes the sum of squared residuals over all equations. For uncorrelated, zero-mean errors of equal variance, this gives the best linear unbiased estimator (BLUE) of $\vx$.

Note that the optimal solution to \eqref{eq:overdetermined_noisy} in the sense
of the minimum mean-squared error (MMSE) (when $\ve$ is considered random)
is not given by the MPP, but rather as the Wiener filter \cite{Kay:1998wba},

\begin{equation}
  \label{eq:mmse_x}
  \mB_{\MMSE} = \mC_x \mA^\H (\mA \mC_x \mA^\H + \mC_n)^{-1} \, ,
\end{equation}
where $\mC_x$ and $\mC_n$ are signal and noise covariance matrices.\footnote{Assuming $\mA \mC_x \mA^\H + \mC_n$ is invertible.} MPP for fat matrices with full row rank is a special case of this formula for $\mC_n = \mat{0}$ and $\mC_x = \mI$.

In the underdetermined case, $\mA \in \C^{m \times n}, \ m < n$, applying the
MPP yields the solution with the smallest $\ell^2$ norm among all vectors
$\vx$ satisfying $\vy = \mA \vx$ (among all \emph{admissible} $\vx$). That is,
\begin{equation}
  \norm{\mA^\dagger \mA \vx}_2 \leq \norm{\vz}_2,  \quad \forall \, \vz \ s.t. \ \mA \vz = \mA \vx.
\end{equation}

To see this, we use the orthogonality of $\mA^\dagger \mA$. Note that any
vector $\vx$ can be decomposed as
\begin{equation}
  \mA^{\dagger} \mA \vx  +  (\mI - \mA^\dagger \mA) \vx,
\end{equation}
and that
\begin{equation}
\begin{aligned}
    \inprod{\mA^{\dagger} \mA \vx, (\mI - \mA^\dagger \mA) \vx}
    &= \inprod{\mA^\dagger \mA \vx, \vx} - \inprod{\mA^\dagger \mA \vx, \mA^\dagger \mA \vx} \\
    &= \inprod{\mA^\dagger \mA \vx, \vx} - \inprod{\mA^\dagger \mA \vx, (\mA^\dagger \mA)^\H \vx} \\
    &= \inprod{\mA^\dagger \mA \vx, \vx} - \inprod{\mA^\dagger \mA \mA^\dagger \mA \vx, \mA^\dagger \mA \vx} \\
    &= \inprod{\mA^\dagger \mA \vx, \vx} - \inprod{\mA^\dagger \mA \vx, \vx} \\
    &= 0,
\end{aligned}
\end{equation}
where we applied Definition \ref{def:MPP} twice. Thus $\mA^\dagger \mA \vx$ is orthogonal to $(\mI - \mA^\dagger \mA) \vx$, and we have
\begin{equation}
  \begin{aligned}
  \norm{\vz}_2^2 &= \norm{\mA^\dagger \mA \vz + (\mI - \mA^\dagger \mA)\vz}_2^2 \\
  &= \norm{\mA^\dagger \mA \vz}_2^2 + \norm{(\mI - \mA^\dagger \mA) \vz}_2^2 \\
  &\geq \norm{\mA^\dagger \mA \vz}^2_2 = \norm{\mA^\dagger \mA \vx}^2_2. 
  \end{aligned}
\end{equation}

A natural question to ask is if there are other MPP-like linear generalized
inverses for $\ell^p$ norms with $p \neq 2$. The answer is negative: 

\begin{proposition}[Corollary 5, \cite{Newman:1969eg}]
    \label{prop:only_MPP_is_linear}
    Let $m \geq 3$ and $n > m$. For $1 < p < \infty$ define $B_\mA : \C^m \to \C^n$ as
    \[
        B_\mA(\vy) \bydef \argmin_{\vz \in \C^n: \mA \vz = \vy} \norm{\vz}_p,
    \]
    where the minimizer is unique by the strict convexity of $\norm{\, \cdot \,}_p$. Then $B_\mA(\vy)$ is linear for all $\mA$ if and only if $p = q = 2$.
\end{proposition}


\section{Generalized Inverses Minimizing Matrix Norms}
\label{sec:normmin}

Despite the negative result in Proposition \ref{prop:only_MPP_is_linear} saying that the MPP is in some sense an exception, an interesting way of generating different generalized inverses is by norm\footnote{For brevity, we loosely call ``norm''  any quasi-norm such as $\ell^{p}$, $p<1$, as well as the ``pseudo-norm'' $\ell^{0}$.} minimization. Two central definitions of such generalized inverses will be used in this paper. The generalized inverse of $\mA \in \C^{m \times n}$, $m < n$, with
minimal $\nu$-norm is defined as ($\norm{\,\cdot\,}_\nu$ is an arbitrary
matrix norm or quasi-norm)
\begin{align}
  \ginv{\nu}{\mA} &\bydef \argmin_{\mX} \ \norm{\mX}_{\nu} \ \ \text{subject to} \ \ \mX \in \ginvset(\mA). \nonumber
\end{align}
The generalized inverse minimizing the $\mu$-norm of the product $\mX \mA$ is
defined as
\begin{align}
  \pginv{\mu}{\mA} &\bydef \argmin_{\mX} \ \norm{\mX \mA}_{\mu} \ \ \text{subject to} \ \ \mX \in \ginvset(\mA). \nonumber
\end{align}
This definition, which is a particular case of the first one with $\norm{\cdot}_{\nu} = \norm{\cdot \mA}_{\mu}$, 
will serve when considering $\mX$ as a poor man's linear replacement for $\ell^{p}$ minimization.
Strictly speaking, the above-defined pseudoinverses are sets, since the corresponding programs may have more than one solution. We will point out the cases when special care must be taken. Another important point is that both definitions involve convex programs. So, at least in principle, we can find the optimizer in the sense that any first-order scheme will lead to the global optimum.


We will treat several families of matrix norms. A matrix norm is any norm on
$\C^{m \times n}$.

\subsection{Entrywise norms} 

The simplest matrix norm is the entrywise $\ell^p$ norm. It is defined through
an isomorphism between $\C^{m \times n}$ and $\C^{mn}$, that is, it is
simply the $\ell^{p}$ norm of the vector of concatenated columns.

\begin{definition}
The $p$-entrywise norm of $\mM \in \C^{m \times n}$, where $0 \leq p \leq \infty$, is given as
\begin{equation}
  \norm{\mM}_p \bydef \norm{\mathrm{vec}(\mM)}_p.
\end{equation}
\end{definition}
A particular entrywise norm is the Frobenius norm associated to $p=2$. 

\subsection{Induced norms---\emph{poor man's $\ell^{p}$ minimization}} 

An important class is that of induced norms. To define these norms, we
consider $\mM \in \C^{m \times n}$ as an operator mapping vectors from $\C^n$
(equipped with an $\ell^{p}$ norm) to $\C^m$ (equipped with an $\ell^{q}$
norm).

\begin{definition}
The $\lplq{p}{q}$ induced norm of $\mM \in \C^{m \times n}$, where $0 < p,q \leq \infty$ is 
\begin{equation}
  \normlplq{p}{q}{\mM} \bydef \sup_{\vx \neq 0} \frac{\norm{\mM \vx}_q}{\norm{\vx}_p}.
\end{equation}
\end{definition}
It is straightforward to show that this definition is equivalent to
$\normlplq{p}{q}{\mM} = \sup_{\norm{\vx}_p=1} \norm{\mM \vx}_q$. Note that
while this is usually defined only for proper norms (i.e., with $1 \leq p,q
\leq \infty$) the definition remains valid when $0 < p < 1$ and/or $0<q < 1$.

\subsection{Mixed norms (columnwise and rowwise)} 

An interesting case mentioned in the introduction is the $\ell^1 \rightarrow
\ell^1$ induced norm of $\mX \mA$, as it leads to a sort of optimal poor man's
$\ell^1$ minimization. The $\ell^{1} \rightarrow \ell^{1}$ induced norm is a
special case of the family of $\ell^1 \rightarrow
\ell^q$ induced norms, which can be shown to have a simple expression as columnwise mixed norm
\begin{equation}
\label{eq:l1lq_characterization}
\normlplq{1}{q}{\mM} = \max_{1 \leq j \leq n} \norm{\vm_{j}}_{q} \bydef \normcolpq{q}{\infty}{\mM}. 
\end{equation}

More generally, one can consider columnwise mixed norms for any $p$ and $q$:

\begin{definition}
  The columnwise mixed norm $\normcolpq{p}{q}{\mM}$ is defined as
  \begin{equation}
  \normcolpq{p}{q}{\mM} \bydef \bigg(\sum_{j} \norm{\vm_{j}}_{p}^q\bigg)^{1/q}
  \end{equation}
  with the usual modification for $q=\infty$.
\end{definition}

We deal both with column- and row-wise norms, so we introduce a mnemonic
notation to easily tell them apart. Thus $\normcolpq{p}{q}{\ \cdot \ }$
denotes columnwise mixed norms, and $\normrowpq{p}{q}{\ \cdot \ }$ denotes
rowwise mixed norms, defined as follows.

\begin{definition}
  The rowwise mixed norm $\normrowpq{p}{q}{\mM}$ is defined as
  \begin{equation}
    \normrowpq{p}{q}{\mM} \bydef \bigg(\sum_{i} \norm{\vm^{j}}_{p}^q\bigg)^{1/q}
  \end{equation}
  with the usual modification for $q=\infty$.
\end{definition}

\subsection{Schatten norms} 

Another classical norm is the spectral norm, which is the $\ell^{2}
\rightarrow \ell^{2}$ induced norm. It equals the maximum singular value of
$\mM$, so it is also a special case of Schatten norm, just as the Frobenius norm
which is the $\ell^{2}$ norm of the vector of singular values of $\mM$. We can
also define a general Schatten norm $\norm{\singvals{\mM}}_{p}$ where
$\singvals{\mM}$ is the vector of singular values.

\begin{definition}
  The Schatten norm $\norm{\mM}_{S^p}$ is defined as
  \begin{equation}
    \norm{\mM}_{S^p} \bydef \norm{\singvals{\mM}}_p,
  \end{equation}
  where $\singvals{\mM}$ is the vector of singular values.
\end{definition}

As we will see further on, these are special cases of the larger class of
unitarily invariant matrix norms.

\subsection{Poor man's $\ell^p$ minimization revisited}
\label{sub:poor_mans_lp_minimization}

%
%

We conclude the overview of matrix norms by introducing certain norms based on
a probabilistic signal model. Similarly to induced norms on the projection
operator $\mX \mA$, these norms lead to optimal $\ell^p$-norm blow-up that can
be achieved by a linear operator. Given that they are computed on the
projection operator, they are primarily useful when considering $\pginv{}{\cdot}$, not
$\ginv{}{\cdot}$.

We already pointed out in the introduction that the generalized inverse
$\pginv{\lplq{p}{p}}{\mA}$ is the one which minimizes the worst case blowup of
the $\ell^p$ norm between $\vz$, the minimum $\ell^p$ norm vector such that
$\vy = \mA \vz$, and the linear estimate $\mX \vy$ where $\mX \in
\ginvset(\mA)$. In this sense, $\pginv{\lplq{p}{p}}{\mA}$ provides the best
worst-case poor man's (linear) $\ell^{p}$ minimization, and solves
\begin{equation}
\inf_{\mX \in \ginvset(\mA)} \sup_{\vy \neq 0} \frac{\norm{\mX\vy}_{p}}{\inf_{\vz: \mA\vz=\vy} \norm{\vz}_{p}}. 
\end{equation}
In this expression we can see explicitly the \emph{true} $\ell^p$ optimal
solution in the denominator of the argument of the supremum.

Instead of minimizing the worst-case blowup, we may want to minimize
average-case $\ell^p$ blowup over a given class of input vectors. Let $\vu$ be
a random vector with probability distribution given by $\mathbb{P}_{\vu}$.
Given $\mA$, our goal is to minimize $\E_{\vu \sim \mathbb{P}_{\vu}}[\norm{\mX
\mA
\vu}_p]$. We replace this minimization by a simpler proxy: we minimize
$(\E_{\vu \sim \mathbb{P}_{\vu}}\, \norm{\mX \mA \vu}_p^p)^{\frac{1}{p}}$. It is not difficult to verify
that this expectation defines a norm (or a semi-norm, depending on $\mathbb{P}_{\vu}$).

Interestingly, for certain distributions $\mathbb{P}_{\vu}$ this leads back to
minimization of standard matrix norms:

\begin{proposition}\label{prop:avg_case_minimization}
  Assume that $\vu \sim \mathcal{N}(\vec{0}, \mI_n)$. Then we have 
  \begin{equation}
    \label{eq:avg_case_minimization}
    \argmin_{\mX \in \ginvset(\mA)} \left(\E_{\vu}[\norm{\mX \mA \vu}_p^{p}]\right)^{1/p} = \pginv{\rowpq{2}{p}}{\mA}
  \end{equation}
\end{proposition}

\begin{remark}
  This result is intuitively satisfying. It is known \cite{Jenatton:2011vd}
  that the $\rowpq{2}{1}$ mixed norm promotes row sparsity, thus the resulting
  $\mX$ will have rows set to zero. Therefore, even if the result may have
  been predicted, it is interesting to see how a generic requirement to have a
  small $\ell^1$ norm of the output leads to a known group sparsity penalty on the product matrix.
\end{remark}

\begin{remark}
    For a general $p$ the minimization \eqref{eq:avg_case_minimization} only
    minimizes an expected \emph{proxy} of the output $\ell^p$ norm, but for
    $p=1$ we get exactly the expected $\ell^1$ norm.
\end{remark}

\begin{proof}
\begin{equation}
\begin{aligned}
 \E_{\vu} \big[\norm{\mX \mA \vu}_p^p \big] 
  &= 
 \sum_{i=1}^n \E_{\vu} \big[\abs{(\mX \mA \vu)_i}^p \big]
\end{aligned}
\end{equation}
Because $\vu$ is centered normal with covariance $\mI_n$, the covariance
matrix of $\mX \mA \vu$ is $\mK = (\mX \mA)(\mX \mA)^\H$. Individual
components are distributed according to $(\mX \mA \vu)_i \sim \mathcal{N}(0,
\mK_{ii}) = \mathcal{N}(0, \norm{\vx^i\mA}_2^2)$. A straightforward
computation shows that
\begin{equation}
  \E \big[ \abs{(\mX\mA\vu)_i}^p \big] = \frac{2^{p/2} \, \Gamma\left(\frac{1+p}{2}\right)}{\sqrt{\pi}} \norm{\vx^i \mA}_2^p.
\end{equation}
We can then continue writing
\begin{equation}
\begin{aligned}
  \E \big[\norm{\mX \mA \vu}_p^{p} \big] 
  &= 
  \sum_{i=1}^n \norm{\vx^i \mA}_2^p
  = \normrowpq{2}{p}{\mX \mA},
\end{aligned}
\end{equation}
and the claim follows.
\end{proof}


\subsection{Fat and skinny matrices}
\label{sub:fat_and_skinny}

In this paper, we concentrate on generalized inverses of fat
matrices---matrices with more columns than rows. We first want to show that
there is no loss of generality in making this choice. This is clear for
minimizing mixed norms and Schatten norms, as for mixed norms we have that
\begin{equation}
  \normcolpq{p}{q}{\mM} = \normrowpq{p}{q}{\mM^\H},
\end{equation}
and for Schatten norm we have
\begin{equation}
  \norm{\mM}_{S_p} = \norm{\mM^\H}_{S_p}.
\end{equation}
It only remains to be shown for induced norms. We can state the following lemma:
\begin{lemma}
  Let $1 \leq p, q, p^{\H},q^{\H}  \leq \infty$ with $\tfrac{1}{p}+\tfrac{1}{p^{\H}}=1$ and
$\tfrac{1}{q}+\tfrac{1}{q^{\H}}=1$. Then we have the relation
  \normlplq{p}{q}{\mM} = \normlplq{q^*}{p^*}{\mM^\H}.
\end{lemma}
In other words, all norms we consider on tall matrices can be converted
to norms on their fat transposes, and our results apply accordingly.

As a corollary we have
\begin{equation}
  \label{eq:lqlinf_characterization}
  \normlplq{p}{\infty}{\mM} =
  \normlplq{1}{p^{\H}}{\mM^{\H}} = \normcolpq{p^{\H}}{\infty}{\mM^{\H}} \bydef \normrowpq{p^{\H}}{\infty}{\mM} = \max_{1 \leq i \leq m} \norm{\vm^{i}}_{p^{\H}}. 
\end{equation}

Next, we show that generalized inverses obtained by minimizing columnwise
mixed norms always match minimizing an entrywise norm or an induced norm.
\begin{lemma}
  \label{lem:ColumnEntrywise}
  Consider $0 < p \leq \infty$ and a full rank matrix $\mA$. 
  \begin{enumerate}
  \item For $0< q < \infty$, we have the set equality $\ginv{\colpq{p}{q}}{\mA} = \ginv{p}{\mA}$.
  \item For $q = \infty$ we have $\normcolpq{p}{\infty}{\,\cdot\,} = \normlplq{1}{p}{\,\cdot\,}$ and the set inclusions / equalities
  \begin{eqnarray*}
    \ginv{p}{\mA} \subset \ginv{\colpq{p}{\infty}}{\mA} &=& \ginv{\lplq{1}{p}}{\mA}\\
     \pginv{\colpq{p}{\infty}}{\mA} &=& \pginv{\lplq{1}{p}}{\mA}.
  \end{eqnarray*}
  \end{enumerate}
\end{lemma}

\begin{proof}
  For $q<\infty$, minimizing $\normcolpq{p}{q}{\mX}$ under the constraint $\mX
  \in \ginvset(\mA)$ amounts to minimizing $\sum_{j} \norm{\vx_{j}}_{p}^{q}$
  under the constraints $\mA \vx_{j} = \ve_{j}$, where $\vx_{j}$ is the $j$th
  column of $\mX$ and $\ve_{j}$ the $j$th canonical vector. Equivalently, one
  can separately minimize $\norm{\vx_{j}}_{p}$ such that $\mA \vx_{j} =
  \ve_{j}$.
\end{proof}


\subsection{Unbiasedness of Generalized Inverses}\label{sub:unbiasedness}

Most of the discussion so far involved deterministic matrices and deterministic properties. In the previous subsections certain results were stated for matrices in general positions---a property which matrices from the various random ensembles verify with probability one. In this and the next section we discuss properties generalized inverses of some random matrices. We start by demonstrating a nice property of the MPP of random matrices replicated by other norm-minimizing generalized inverses, which is unbiasedness in a certain sense. 

For a random Gaussian matrix $\mA$ it holds that
\begin{equation}
  \frac{n}{m}\E[\mA^\dagger \mA] = \mI.
\end{equation}
In other words, for this random matrix ensemble, applying $\mA$ to a vector, and then the MPP to the measurements will on average retrieve the scaled version of the input vector. This aesthetically pleasing property is used to make statements about various iterative algorithms such as iterative hard thresholding. To motivate it we consider the following generic procedure: let $\vy = \mA \vx$, where $\vx$ is the object we are interested in (e.g. an image), $\mA$ a dimensionality-reducing measurement system, and $\vy$ the resulting measurements. One admissible estimate of $\vx$ is given by $\mX \mA \vx$, where $\mX \in \ginvset(\mA)$. If the dimensionality-reducing system is random, we can compute the expectation of the reconstructed vector as
\begin{equation}
    \E[\mX \vy] = \E[\mX \mA \vx] = \E[\mX \mA] \vx.
\end{equation}
Provided that $\E[\mX \mA] = \frac{m}{n}\mI$, we will obtain, on average, a scaled version of the object we wish to reconstruct.%
\footnote{This linear step is usually part of a more complicated algorithm which also includes a nonlinear \emph{denoising} step (\emph{e.g., thresholding, non-local means}). If this denoising step is a contraction in some sense (\emph{i.e.} it brings us closer to the object we are reconstructing), the following scheme will converge: $\vx^{(k+1)} \bydef \eta(\vx^{(k)} + \frac{n}{m} \mX(\vy - \mA \vx^{(k)}))$.}

Clearly, this property will not hold for generalized inverses obtained by inverting a particular minor of the input matrix. As we show next, it does hold for a large class of norm-minimizing generalized inverses.

\begin{theorem}
  \label{thm:unbiased}
  Let $\mA \in \R^{m \times n},~m<n$ be a random matrix with iid columns such
  that $a_{ij} \sim (-a_{ij})$. Let further $\norm{\,
  \cdot \, }_{\nu}$ be any matrix norm such that $\norm{\mPi \, \cdot}_{\nu} =
  \norm{\, \cdot \, }_{\nu}$ and $\norm{\mSigma \, \cdot}_{\nu} = \norm{\, \cdot
  \, }_{\nu}$, for any permutation matrix $\mPi$ and modulation matrix
  $\mSigma = \diag(\vsigma)$, where $\vsigma \in \set{-1, 1}^n$. Then, if $\ginv{\nu}{\mA}$ and $\pginv{\nu}{\mA}$ are singletons for all $\mA$, we have
  \begin{equation}
    \E [ \ginv{\nu}{\mA} \mA ] = \E [ \pginv{\nu}{\mA} \mA ] = \tfrac{m}{n}\mI_n
  \end{equation}     
More generally, consider a function $f: C \mapsto f(C) \in C \subset \R^{n
\times m}$ that selects a particular representative for any bounded convex set
$C$, and assume that $f(\mU C) = \mU f(C)$ for any unitary matrix $\mU$ and
any $C$. Examples of such functions $f$ include selecting the centroid of the
convex set, or selecting its element with minimum Frobenius norm.
 We have
  \begin{equation}
    \E [ f(\ginv{\nu}{\mA}) \mA ] = \E [ f(\pginv{\nu}{\mA})   \mA ] = \tfrac{m}{n}\mI_n.
  \end{equation}
\end{theorem}

\begin{remark}
  This includes all classical norms (invariance to row permutations and sign
  changes) as well as any left-unitarily invariant norm (permutations and sign
  changes are unitary).
\end{remark}

To prove the theorem we use the following lemma,

\begin{lemma}\label{lem:ginvU}
  Let $\mU \in \C^{n \times n}$ be an invertible matrix, and $\norm{\, \cdot\, }_{\nu}$ a norm such that $\norm{\mU\, \cdot\, }_{\nu} =  \norm{\, \cdot\, }_{\nu}$. Then the following claims hold,
  \begin{align}
    &\ginv{\nu}{\mA \mU} = \mU^{-1} \ginv{\nu}{\mA},  \\
    &\pginv{\nu}{\mA \mU} = \mU^{-1} \, \pginv{\nu}{\mA}
  \end{align}
  for any $\mA$.
\end{lemma}

\begin{proof}[Proof of the lemma]
  We only prove the first claim; the remaining parts follow analogously using that $\pginv{\nu}{\mA} = \ginv{\mu}{\mA}$ where $\norm{\,\cdot\,}_{\mu} \bydef \norm{\,\cdot\, \mA}_{\nu} = \norm{\mU \,\cdot\, \mA}_{\nu} = \norm{\mU \,\cdot\,}_{\mu}$.

  \noindent {\em Feasibility:} $(\mA \mU) (\mU^{-1} \ginv{\nu}{\mA}) = \mA
    \ginv{\nu}{\mA} = \mI_m$.

  \noindent {\em Optimality:} Consider any $\mX \in \ginvset(\mA\mU)$.
  Since
  $
     (\mA  \mU) \mX = \mI_m = \mA (\mU \mX),
  $
  the matrix  $\mU \mX$ belongs to $\ginvset(\mA)$ hence
   \begin{align}
     \norm{ \mX}_{\nu}
      = \norm{\mU \mX}_{\nu}
      \geq \norm{\ginv{\nu}{\mA }}_{\nu} 
      = \norm{\mU \mU^{-1} \ginv{\nu}{\mA }}_{\nu}
      = \norm{\mU^{-1} \ginv{\nu}{\mA }}_{\nu}.
  \end{align}
\end{proof}

\begin{proof}[Proof of the theorem]
  Since the matrix columns are iid, $\mA$ is distributed identically to $\mA
  \mPi$ for any permutation matrix $\mPi$. This implies that functions of $\mA$ and
  $\mA \mPi$ have the same distribution. Thus the sets $f(\ginv{\nu}{\mA}) \mA$ and $f(\ginv{\nu}{\mA\mPi}) \mA\mPi$ are identically distributed. 
  Using Lemma~\ref{lem:ginvU} with $\mU=\mPi$, we have that
  \begin{align}
    \mM \bydef \E [f(\ginv{\nu}{\mA}) \mA]
    &= \E [f(\ginv{\nu}{\mA\mPi}) \mA\mPi] \\
    & = \E [f(\mPi^{\H}\ginv{\nu}{\mA}) \mA\mPi] \\
    & = \E [\mPi^{^\H} f(\ginv{\nu}{\mA}) \mA \mPi] = \mPi^{\H} \mM \mPi.
  \end{align}

This is more explicitly written $m_{ij} = m_{\pi(i)\pi(j)}$ for all $i,j$ and $\pi$ the permutation associated to the permutation matrix $\mPi$. Since this holds for any permutation matrix, 
   we can write

  \begin{equation} 
  \mM =
    \begin{bmatrix}
      c & b &  \cdots & b \\
      b & c &  \cdots & b \\
      \vdots & & \ddots & \vdots \\
      b & b & \cdots & c
    \end{bmatrix},
  \end{equation}
  We compute the value of $c = \tfrac{m}{n}$ as follows:

   \begin{align}
    nc = \trace \ \E[f(\ginv{\nu}{\mA}) \mA]
    &= \E[ \trace \ f(\ginv{\nu}{\mA}) \mA] 
    = \E[ \trace \ \mA \, f(\ginv{\nu}{\mA})]
    = \trace \ \mI_m = m.
  \end{align}

To show that $b=0$, we observe that since $a_{ij} \sim (-a_{ij})$, the matrices $\mA$ and $\mA \mSigma$ have the same distribution. As above, using Lemma~\ref{lem:ginvU} with $\mU = \mSigma$, this implies $\mM = \mSigma \mM \mSigma$ for any modulation matrix $\mSigma$, that is to say $m_{ij} = \sigma_{i}m_{ij}\sigma_{j}$ for any $i,j$ and $\vsigma \in \{-1,1\}^{n}$. It follows that $m_{ij}=0$ for $i \neq j$.
  Since we already established that $c_i \equiv \frac{m}{n}$ we conclude that $\E
  [f(\ginv{\nu}{\mA})\mA] = \frac{m}{n} \mI_n$.

\end{proof}

The conditions of Theorem \ref{thm:unbiased} are satisfied by various random matrix ensembles including the common iid Gaussian ensemble.

One possible interpretation of this result is as follows: For the
Moore-Penrose pseudoinverse $\mA^\dagger$ of a fat matrix $\mA$, we have that
$\mA^\dagger \mA$ is an orthogonal projection. For a general pseudoinverse
$\mX$, $\mX \mA$ is an oblique projection, along an angle different than
$\frac{\pi}{2}$. Nevertheless, for many norm minimizing generalized inverses
this angle is on average $\frac{\pi}{2}$, if the average is taken over common
classes of random matrices.


\section{
Norms Yielding the Moore-Penrose Pseudoinverse}
\label{sec:normsmpp}
A particularly interesting property of the MPP is that it minimizes many of
the norms in Table~\ref{tab:norm-list}. This is related to their
unitary invariance, and to geometric interpretation of the MPP \cite{Vetterli:2014vm}.

\subsection{Unitarily invariant norms}

\begin{definition}[Unitarily invariant matrix norm]
  A matrix norm $\norm{\cdot}$ is called unitarily invariant if and only if
  $\norm{\mU \mM \mV} = \norm{\mM}$ for any $\mM$ and any unitary matrices $\mU$ and
  $\mV$.
\end{definition}
Unitarily invariant matrix norms are intimately related to symmetric gauge
functions \cite{Mirsky:1960fi}, defined as vector norms invariant to sign
changes and permutations of the vector entries. A theorem by Von Neumann
\cite{vonNeumann:1937uj, Horn:2012tf} states that any unitarily invariant norm
$\norm{\cdot}$ is a symmetric gauge function $\phi$ of the singular values, i.e.,
$\norm{\cdot}=\phi(\sigma(\cdot))
\bydef \norm{\cdot}_{\phi}$. To be a symmetric gauge function, $\phi$ has to
satisfy the following properties \cite{Mirsky:1960fi}: 
\begin{enumerate}[label=(\roman*)]
  \item $\phi(\vx) \ge 0$ for $\vx \neq 0$,
  \item $\phi(\alpha \vx) = \abs{\alpha} \phi(\vx)$,
  \item $\phi(\vx + \vy) \leq \phi(\vx) + \phi(\vy)$,
  \item $\phi(\mPi \vx) = \phi(\vx)$,
  \item $\phi(\mSigma \vx) = \phi(\vx)$,
\end{enumerate}
where $\alpha \in \R$, $\mPi$ is a permutation matrix, and $\mSigma$ is a
diagonal matrix with diagonal entries in $\{-1,+1\}$. Zi\c{e}tak \cite{Zietak:1997ut} shows that the MPP minimizes any unitarily
invariant norm.
\begin{theorem}[Zi\c{e}tak, 1997]\label{th:Zietak}
  Let $\norm{\cdot}_{\phi}$ be a
  unitarily invariant norm corresponding to a symmetric gauge function
  $\phi$. Then, for any $\mA \in \C^{m \times n}$, $\norm{\mA^{\dagger}}_{\phi} = \min \,
  \set{\norm{\mB}_{\phi}: \mB \in \ginvset(\mA)}$. If additionally
  $\phi$ is strictly monotonic, then the set of minimizers contains a
  single element $\mA^{\dagger}$.
\end{theorem}

It is interesting to note that in the case of the operator norm, which is
associated to the symmetric gauge function $\phi(\cdot) = \norm{\cdot}_{\infty}$,
the minimizer is not unique. Zi\c{e}tak mentions a simple example for
rank-deficient matrices, but multiple minimizers are present in the full-rank
case too, as is illustrated by the following example.

\begin{example}\label{ex:MPPnonUniqueOpNorm}
  Let the matrix $\mA$ be
  \begin{equation}
    \mA = 
    \begin{bmatrix}
      1 & 1 & 0 \\
      1 & 0 & 1
    \end{bmatrix}.
  \end{equation}
  Singular values of $\mA$ are $\sigma_1=\sqrt{3}$ and $\sigma_2=1$, and its MPP is
  \begin{equation}
    \mA^{\dagger} =
    \mV
    \begin{bmatrix}
      \frac{\sqrt{3}}{3} & 0 \\
      0 & 1 \\
      0 & 0
    \end{bmatrix}
    \mU^\H
    = \frac{1}{3}
    \begin{bmatrix}
      1 & 1 \\
      2 & -1 \\
      -1 & 2
    \end{bmatrix}.
  \end{equation}
  Consider now matrices of the form
  \begin{equation}
    \mA^{\ddagger} =
    \mV
    \begin{bmatrix}
      \frac{\sqrt{3}}{3} & 0 \\
      0 & 1 \\
      \alpha & 0
    \end{bmatrix}
    \mU^\H.
  \end{equation}
  It is readily verified that $\mA \mA^{\ddagger} = \mI$ and $\sigma(\mA^{\ddagger}) = \set{\sqrt{\alpha^2
    + \frac{1}{3}},\ 1}$. Hence, whenever $0 < \abs{\alpha} \leq \sqrt{\frac{2}{3}}$, we have that
  $\norm{\mA^{\ddagger}}_{S_{\infty}} = \norm{\sigma(\mA^{\ddagger})}_{\infty} = 1 = \norm{\mA^{\dagger}}_{S_{\infty}}$, and yet $\mA^{\ddagger} \neq
  \mA^{\dagger}$.
\end{example}

\subsection{Left unitarily invariant norms}
Another case of particular interest is when the norm is not fully unitarily
invariant, but it is still unitarily invariant on one side. As we have
restricted our attention to fat matrices, we will examine left unitarily
invariant norms, because these will conveniently again lead to the MPP. 
%
%
%

\begin{lemma}
  \label{lem:luinv-block}
  Let $\mX, \mY \in \C^{n \times m}$ be defined as
  \begin{equation}
    \mX =
    \begin{bmatrix}
       \mK_1 \\ \mat{0}
    \end{bmatrix},
    \quad
    \mY = 
    \begin{bmatrix}
         \mK_1 \\ \mK_2
    \end{bmatrix}.
  \end{equation}
  Then $\norm{\mX} \leq \norm{\mY}$ for any left unitarily invariant norm
  $\norm{\cdot}$.
\end{lemma}

\begin{proof}
Observe that $\mX = \mT \mY$, where
  \begin{equation}
\mT = 
\begin{bmatrix}
\mI & \mat{0} \\ 
\mat{0} & \mat{0}
\end{bmatrix}
    = \frac{1}{2}
    \begin{bmatrix}
      \mI & \mat{0} \\
      \mat{0} & -\mI
    \end{bmatrix}
    + \frac{1}{2}
    \begin{bmatrix}
      \mI & \mat{0} \\
      \mat{0} & \mI
    \end{bmatrix},
  \end{equation}
  with the two matrices on the right-hand side being unitary. The claim follows by applying the triangle inequality.
\end{proof}

With this lemma in hand, we can prove the main result for left unitarily invariant norms.

\begin{theorem}\label{th:LUINorms}
  Let $\norm{\cdot}$ be a left unitarily invariant norm, and let $\mA \in \C^{m \times n}$ be full rank with $m < n$. Then
  $\norm{\mA^{\dagger}} = \min \ \set{\norm{\mX}: \mX \in \ginvset(\mA)}$. If $\norm{\cdot}$ satisfies a strict inequality in Lemma \ref{lem:luinv-block} whenever $\mK_2 \neq \mat{0}$, then the set of minimizers is a singleton $\set{\mA^{\dagger}}$.
\end{theorem}
\begin{proof}
Write $\mY \in \ginvset(\mA)$ in the form~\eqref{eq:GinvsetMform}. By the left unitary invariance and Lemma~\ref{lem:luinv-block}
\begin{eqnarray*}
\norm{\mY} = \norm{\mM \mU^{\H}} = \norm{\left[ \begin{array}{c}\mSigma_{\square}^{-1} \mU^{\H}\\  \mS \mU^{\H}\end{array} \right]} \geq \norm{\left[ \begin{array}{c}\mSigma_{\square}^{-1} \mU^{\H}\\  \mat{0} \end{array} \right]} = \norm{\mA^{\dagger}}.
\end{eqnarray*}
\end{proof}

\subsection{Left unitarily invariant norms on the product operator}
An immediate consequence of Theorem \ref{th:LUINorms} is that the same phenomenon occurs when minimizing a left unitarily invariant norm of the product $\mX\mA$. This simply comes from the observation that if $\norm{\cdot}_{\mu}$ is left unitarily invariant, then so is $\norm{\cdot}_{\nu} = \norm{\cdot \mA}_{\mu}$.

\begin{corollary}\label{cor:PLUINorms}
Let $\norm{\cdot}$ be a left unitarily invariant norm, and let $\mA \in \C^{m \times n}$ be full rank with $m < n$. Then $\norm{\mA^{\dagger}} = \min \ \set{\norm{\mX\mA}: \mX \in \ginvset(\mA)}$. If $\norm{\cdot}$ satisfies a strict inequality in Lemma \ref{lem:luinv-block} whenever $\mK_2 \neq \mat{0}$, then the set of minimizers is a singleton $\set{\mA^{\dagger}}$.
\end{corollary}

\subsection{Classical norms leading to the MPP}

As a corollary, some large families of norms lead to MPP. In particular the following holds:
\begin{corollary} 
\label{cor:classicalMPP}
Let $\mA \in \C^{m \times n}$ be full rank with $m < n$. 
\begin{enumerate}
\item {\bf Schatten norms:}\label{it:schattenMPP} for $1 \leq p \leq \infty$ 
\[
\mA^{\dagger}  \in \ginv{S_p}{\mA}
\qquad
\mA^{\dagger}  \in \pginv{S_p}{\mA}.
\]
The considered sets are singletons for $p<\infty$.\\ The set $\pginv{S_\infty}{\mA}$ is a singleton, but
$\ginv{S_\infty}{\mA}$ is \emph{not necessarily} a singleton.
\item {\bf Columnwise mixed norms:} \label{it:colMPP} for $1 \leq q \leq \infty$
\[
\mA^{\dagger}  \in \ginv{\colpq{2}{q}}{\mA}
\qquad
\mA^{\dagger}  \in \pginv{\colpq{2}{q}}{\mA}
\]
The considered sets are singletons for $q<\infty$, but \emph{not always} for $q=\infty$.
\item {\bf Induced norms:} \label{it:inducedMPP} for $1 \leq p \leq \infty$
\[
\mA^{\dagger} \in \ginv{\lplq{p}{2}}{\mA}
\qquad
\mA^{\dagger} \in \pginv{\lplq{p}{2}}{\mA}
\]
The set $\ginv{\lplq{p}{2}}{\mA}$ is \emph{not always} a singleton for $p \leq 2$.
\end{enumerate} 
\end{corollary}
\begin{remark}
Whether $\pginv{\lplq{p}{2}}{\mA}$, and
$\ginv{\lplq{p}{2}}{\mA}$ for $2<p<\infty$ are singletons remains an open
question. 
\end{remark}

\begin{remark}
  Let us highlight an interesting consequence of
  Corollary~\ref{cor:classicalMPP}: consider the $\normlplq{\infty}{2}{\cdot}$
  norm, whose computation is known to be NP-complete \cite{Lewis:2010xx}.
  Despite this fact, Corollary~\ref{cor:classicalMPP} implies that we can find
  an optimal solution of an optimization problem involving this norm.
\end{remark}

Proof of Corollary~\ref{cor:classicalMPP} is given in Appendix~\ref{sub:proof_of_classicalMPP}.

\subsection{Norms that Almost Never Yield the MPP}
\label{sub:norms_not_mpp}

After concentrating on matrix norms whose minimization \emph{always} leads to the MPP, we now take a moment to look at norms that (almost) \emph{never} lead to the MPP. The norms not covered by Corollary~\ref{cor:classicalMPP}: columnwise mixed norms with $p \neq 2$, rowwise norms for $(p,q) \neq (2,2)$,  and induced norms for $q \neq 2$. Our main result is the following theorem whose proof is given in Appendix \ref{sub:proof_of_theorem_not_always_mpp}.
\begin{theorem}
    \label{thm:not_always_mpp}
    Consider $m<n$. 
    \begin{enumerate}
    \item For any $m \geq 1$, $\mA \in \R^{m \times n}$, $0<p\leq \infty$, $0<q<\infty$, $\ginv{\colpq{p}{q}}{\mA} = \ginv{p}{\mA}$. Moreover, there exists $\mA_{1} \in \R^{m \times n}$ such that for $0<p,q<\infty$:
    \begin{eqnarray*}
    \mA_{1}^{\dag} \in \ginv{\colpq{p}{q}}{\mA_{1}}  & \Longleftrightarrow &  p = 2
    \end{eqnarray*}
    \item  For $m=1$,  and any $\mA \in \R^{1 \times n}$, $0<p\leq \infty$, $0<q<\infty$, $\ginv{\rowpq{p}{q}}{\mA} = \ginv{q}{\mA}$. Hence, the matrix $\mA_{1} \in \R^{1 \times n}$ satisfies for $0<p\leq \infty$, $0<q < \infty$:
    \begin{eqnarray*}
    \mA_{1}^{\dag} \in \ginv{\rowpq{p}{q}}{\mA_{1}}  & \Longleftrightarrow &  q = 2
    \end{eqnarray*}
    For any $m \geq 2$, there exists $\mA_{2} \in \R^{m \times n}$ such that for $0<p,q<\infty$: 
    \begin{eqnarray*}
    \mA_{2}^{\dag} \in \ginv{\rowpq{p}{q}}{\mA_{2}} & \Longleftrightarrow &  p = 2
    \end{eqnarray*} 
    For any $m \geq 3$, there exists $\mA_{3} \in \R^{m \times n}$ such that  for $0<p,q<\infty$:
    \begin{eqnarray*}
    \mA_{3}^{\dag} \in \ginv{\rowpq{p}{q}}{\mA_{3}} & \Longleftrightarrow &  (p,q) \in \set{(2,2),(1, 2)}
    \end{eqnarray*}
    Whether a similar construction exists for $m =2$ remains open.
    Combining the above, for any $m \geq 3$,  $0<p,q<\infty$, the following properties are equivalent
    \begin{itemize}
       \item for any $\mA  \in \R^{m \times n}$ we have $\mA^{\dag} \in \ginv{\rowpq{p}{q}}{\mA}$
       \item $(p,q) = (2,2)$.
    \end{itemize}   
    The matrix $(\mA_{3})^{\dag}$ has positive entries. In fact, any $\mA$ such that $(\mA)^{\dag}$ has positive entries satisfies
    \(
    (\mA)^{\dag} \in \ginv{\rowpq{1}{2}}{\mA}.
    \)

    \item For $m=1$, and any $\mA \in \R^{1 \times n}$, $0<p,q\leq \infty$, $\pginv{\colpq{p}{q}}{\mA} = \ginv{p}{\mA}$.
    Hence, the matrix $\mA_{1} \in \R^{m \times 1}$ satisfies for $0<p<\infty$, $0<q\leq\infty$:
    \begin{eqnarray*}
    \mA_{1}^{\dag} \in \ginv{\colpq{p}{q}}{\mA_{1}}  & \Longleftrightarrow &  p = 2
    \end{eqnarray*}

    For any $m \geq 2$,  there exists $\mA_{4} \in \R^{m \times n}$ such that for $1\leq p < \infty$ and $0<q<\infty$:
    \begin{eqnarray*}
    \mA_{4}^{\dag} \in \pginv{\colpq{p}{q}}{\mA_{4}}  & \Longleftrightarrow &  p = 2
    \end{eqnarray*} 
    Whether a similar construction is possible for $m \geq 2$ and $0<p \leq 1$ remains open.
    \item  For $m=1$, and any $\mA \in \R^{1 \times n}$, $0<p,q \leq \infty$, $\pginv{\rowpq{p}{q}}{\mA} = \ginv{q}{\mA}$. Hence, the matrix $\mA_{1} \in \R^{m \times 1}$ satisfies for $0<p \leq \infty$, $0<q <\infty$:
    \begin{eqnarray*}
    \mA_{1}^{\dag} \in \pginv{\rowpq{p}{q}}{\mA_{1}}  & \Longleftrightarrow &  q = 2
    \end{eqnarray*}

    For any $m \geq 2$, the matrix $\mA_{4}$ satisfies for $1 \leq p\leq \infty$, $0<q < \infty$:
    \begin{eqnarray*}
    \mA_{4}^{\dag} \in \pginv{\rowpq{p}{q}}{\mA_{4}}  & \Longleftrightarrow &  q = 2
    \end{eqnarray*}
    Whether a similar construction is possible for $m \geq 2$ and $0<p \leq 1$ remains open.
  
    For any $m \geq 3$, there exists $\mA_{5} \in \R^{m \times n}$ such that  for $1 \leq p<\infty$ and $q=2$:
    \begin{eqnarray*}
        \mA_{5}^{\dag} \in \pginv{\rowpq{p}{q}}{\mA_{3}} & \Longleftrightarrow &  p=2
    \end{eqnarray*} 
    Whether a similar construction for $m =2$ and/or $0 < p \leq 1$ is possible remains open.

    Combining the above, for any $m \geq 3$,  $1 \leq p < \infty$ $0<q<\infty$, the following properties are equivalent
    \begin{itemize}
       \item for any $\mA  \in \R^{m \times n}$ we have $\mA^{\dag} \in \pginv{\rowpq{p}{q}}{\mA}$
       \item $(p,q) = (2,2)$.
    \end{itemize}   
 \end{enumerate}
\end{theorem}

\begin{remark}
It seems reasonable to expect that for ``most'' $p, q$ and ``most'' matrices, the corresponding set of generalized inverses will not contain the MPP. However, the case of rowwise norms $\normrowpq{1}{2}{\,\cdot\,}$ and matrices with positive entries suggests that one should be carefully about the precise statement.
\end{remark}

\section{Matrices Having the Same Inverse for Many Norms}

\label{sec:same_inverse_for_many_norms}

As we have seen, a large class of matrix norms are minimized by the Moore-Penrose pseudoinverse. We now discuss some classes of matrices whose generalized inverses minimize multiple norms.

\subsection{Matrices with MPP whose non-zero entries are constant along columns}

We first look at matrices for which all nonzero-entries of any column of $\mA^{\dagger}$ have the same magnitude. For such matrices, the MPP actually minimizes many norms beyond those already covered by Corollary~\ref{cor:classicalMPP}.

\begin{theorem}\label{th:TFUC}
  Let $\mA \in \C^{m \times n}$. Suppose that every column of $\mX =
  \mA^{\dagger}$ has entries of the same magnitude (possibly differing among
  columns) over its non-zero entries; that is,  $\abs{x_{ij}} \in
  \set{c_j, 0}$. Then the following statements are true:
  \begin{enumerate} 
     \item For $1 \leq p \leq \infty$, $0< q \leq \infty$, we have
     \[
     \mA^{\dagger} \in \ginv{\colpq{p}{q}}{\mA}.
     \] 
     This set is a singleton for $1<p<\infty$ and $0<q<\infty$. \\
     NB: this includes the nonconvex case $0<q<1$.
     
     \item For $1 \leq p \leq \infty$, $1 \leq q \leq \infty$, assuming further that $c_{j}=c$ is identical for all columns, we have
     \[
     \mA^{\dagger} \in \ginv{\rowpq{p}{q}}{\mA}.
     \] 
     This set is a singleton for $1<p<\infty$, $1 \leq q < \infty$.
  \end{enumerate}
\end{theorem}

\begin{example}
    Primary examples of matrices $\mA$ satisfying the assumptions of Theorem~\ref{th:TFUC} are tight frames $\mA$ with entries of constant magnitude, such as partial Fourier matrices, $\mA = \mR \mF$ (resp. partial Hadamard matrices, $\mA = \mR \mH$), with $\mR$ a restriction of the identity matrix $\mI_{n}$ to some arbitrary subset of $m$ rows and $\mF$ (resp. $\mH$) the Fourier (resp. a Hadamard) matrix of size $n$. Indeed, when $\mA$ is a tight frame, we have $\mA \mA^{\H}\ \propto\ \mI_{m}$, hence $\mA^{\dagger} = \mA^\H(\mA \mA^\H)^{-1}\ \propto\ \mA^\H$. When in addition the entries of $\mA$ have equal magnitude, so must the entries of $\mA^{\dagger}$.
\end{example}

The proof of Theorem~\ref{th:TFUC} uses a characterization of the gradient
of the considered norms.

\begin{proof}[Proof of Theorem~\ref{th:TFUC}]
    Consider $\mE$ a matrix such that the matrix $\mY = \mX + \mE$ still
    belongs to $\ginvset(\mA)$, and denote its columns as $\veps_{i}$. For each
    column we have $\mA (\vx_i + \veps_{i}) = \ve_i = \mA \vx_{i}$, that is, $\veps_{i}$ must be in the nullspace of $\mA$ for each column $i$. Since $\mathcal{N}(\mA) = \Span(\mA^\H)^{\perp} = \Span(\mX)^{\perp}$, $\veps_{i}$ must be orthogonal to any column of $\mX$, and in particular to $\vx_{i}$. As a result, for each column 
    \begin{equation} 
     \inprod{\vx_{i},\veps_{i}}_{\C^{n}} = 0\label{eq:orthocols}.
    \end{equation}

    \begin{enumerate}
    \item To show this statement, it suffices to show that the columns $\vx_{i}$ of $\mX = \mA^{\dagger}$ minimize all $\ell^p$ norms, $1 \leq p \leq \infty$, among the columns of $\mY \in \ginvset(\mA)$. 

    First consider $1 \leq p < \infty$ and $f(\vz) \bydef \norm{\vz}_{p}$ for $\vz \in \C^{n}$. By convexity of $\R^{n} \times \R^{n} \ni (\vx_R, \vx_I)  \mapsto f(\vx_R + j \vx_I)$, we have that for any $\vu \in \tfrac{\partial f}{\partial\vz}(\vx_{i})$ (as defined in Lemma~\ref{le:subgradients0})
    \[
    \norm{\vx_{i} + \veps}_p 
    \geq \norm{\vx_{i}}_p + \inprod{\Re(\vu), \veps_R}_{\R^{n}}+ \inprod{\Im(\vu), \veps_I}_{\R^{n}}
    =  \norm{\vx_{i}}_p + \Re\left(\inprod{\vu, \veps}_{\C^{n}}\right),
    \] 
    where we recall that $\inprod{\vu, \veps}_{\C^{n}} = \sum_i u_i \epsilon_i^*$ involves a complex conjugation of $\veps$. By Lemma~\ref{le:subgradients0}, since the column $\vx_i$ of $\mX$ has all its nonzero entries of the same magnitude there is $c = c(p,\vx_{i})>0$ such that $\vu \bydef c \vx_i \in \tfrac{\partial f}{\partial\vz}(\vx_{i})$. In particular, we choose the element of the subdifferential with zeros for entries corresponding to zeros in $\vx_i$.
    Using~\eqref{eq:orthocols} then gives:

    \begin{equation}
    \begin{aligned}
        \norm{\vx_i + \veps_i}_p
        &\geq \norm{\vx_i}_p + c\, \Re\left(\inprod{\vx_{i},\veps_{i}}_{\C^{n}}\right) 
        = \norm{\vx_i}_p.
    \end{aligned}
    \end{equation}
     Since this holds true for any $1 \leq p < \infty$, we also get $\norm{\vx_{i}+\veps_{i}}_{\infty} \geq \norm{\vx_{i}}_{\infty}$ by considering the limit when $p \to \infty$. For $1<p<\infty$ and $0<q<\infty$, the strict convexity of $f$ and the strict monotonicity of the $\ell^{q}$ (quasi)norm imply that the inequality is strict whenever $\mE \neq \mat{0}$, hence the uniqueness result.

    \item First consider $1 \leq p < \infty$ and $1 \leq q < \infty$, and $f(\mZ) \bydef \normrowpq{p}{q}{\mZ}$ for $\mZ \in \C^{n \times m}$. By convexity of $(\mX_R, \mX_I) \in \R^{n \times m} \times \R^{n \times m}\mapsto f(\mX_{R}+j\mX_{I})$ we have for any $\vU \in \tfrac{\partial f}{\partial\mZ}(\mX)$ (as defined in Lemma~\ref{le:subgradients0})
    \[
        \norm{\mX + \mE}_{\rowpq{p}{q}} 
        \geq \norm{\mX}_{\rowpq{p}{q}} + \Re\left(\inprod{\mU, \mE}_{\C^{n \times m}}\right).
    \]
    By Lemma~\ref{le:subgradients0} since $\mX$ has all its nonzero entries of the same magnitude there is a constant $c = c(p,q,\mX)>0$ such that $\mU \bydef c \cdot \mX \in \tfrac{\partial f}{\partial\mZ}(\mX)$ (we again choose the element of the subdifferential with zeros for entries corresponding to zeros in $\mX$). Using~\eqref{eq:orthocols} we have $\inprod{\mX,\mE}_{\C^{n \times m}} = 0$ and it follows that:
    \begin{equation}
    \begin{aligned}
        \norm{\mX + \mE}_{\rowpq{p}{q}}
        &= \norm{\mX}_{\rowpq{p}{q}} + c \cdot \Re\left( \inprod{\mX, \mE}_{\C^{n \times m}}\right)
        = \norm{\mX}_{\rowpq{p}{q}}.
    \end{aligned}
    \end{equation}
    As above, the inequality $\norm{\mX + \mE}_{\rowpq{p}{q}} \geq \norm{\mX}_{\rowpq{p}{q}}$ is extended to $p = \infty$ and/or $q = \infty$ by considering the limit. For $1 < p < \infty$, $1 \leq q < \infty$, the function $f$ is strictly convex implying that the inequality is strict whenever $\mE \neq \mat{0}$. This establishes the uniqueness result. 
    \end{enumerate}
\end{proof}

\subsection{Matrices with a highly sparse generalized inverse} 

Next we consider matrices for which a single generalized inverse, which is {\em not} the MPP, simultaneously minimizes several norms. It is known that if $\vx$ is sufficiently sparse, then it can be uniquely recovered from $\vy = \mA \vx$ by $\ell^p$ minimization with $0 \leq p \leq 1$. 
Denote $k_p(\mA)$ the largest integer $k$ such that this holds true for any $k$-sparse vector $\vx$:
\begin{equation}
  k_p \bydef \max \ \set{k \ : \ \argmin_{\mA\wh{\vx} = \mA \vx} \norm{\wh{\vx}}_p = \{\vx\}, \ \forall \vx \ \text{such that} \ \norm{\vx}_0 \leq k}.
\end{equation}
Using this definition we have the following theorem:
\begin{theorem}\label{th:SparseMatrix}
Consider $\mA \in \C^{m \times n}$, and assume there exists a generalized inverse $\mX \in \ginvset(\mA)$ such that every of its columns $\vx_i$ is $k_{p_{0}}(\mA)$-sparse, for some $0<p_{0} \leq 1$. Then, for all $0 < q < \infty$, and all $0 \leq p  \leq p_{0}$, we have
\[
\ginv{\colpq{p}{q}}{\mA} = \{\mX\}.
\]
\end{theorem}

\begin{proof}
It is known \cite{Gribonval:2007hq} that for $0 \leq p \leq p_{0} \leq 1$ and any $\mA$ we have $k_{p}(\mA) \geq k_{p_{0}}(\mA)$. Hence, the column $\vx_{i}$ of $\mX$ is the unique minimum $\ell^{p}$ norm solution to $\mA \vx = \ve_{i}$. When $q < \infty$ this implies $\ginv{\colpq{p}{q}}{\mA} = \{\mX\}$. When $q=\infty$ this simply implies $\ginv{\colpq{p}{q}}{\mA} \ni \mX$.
\end{proof}

\begin{example}
Let $\mA$ be the Dirac-Fourier (resp. the Dirac-Hadamard) dictionary, $\mA = [\mI,\ \mF]$ (resp. $\mA = [\mI,\ \mH]$). It is known that $k_{1}(\mA) \geq (1+\sqrt{m})/2$ (see, \emph{e.g.}, \cite{Gribonval:2007hq}). Moreoever 
\[
\mX \bydef \begin{bmatrix}
      \mI\\
      \mat{0}
      \end{bmatrix} \in \ginvset(\mA)
\]
has $k$-sparse columns with $k=1$. As a result, $\ginv{\colpq{p}{q}}{\mA} = \{\mX\}$ for any $0 \leq p \leq 1$ and $0<q < \infty$. On this specific case, the equality can also be checked for $q=\infty$. Thus, $\ginv{\colpq{p}{q}}{\mA}$ is distinct from the MPP of $\mA$, $\mA^{\dagger} = \mA^{\H}/2$. 
\end{example}



\section{Computation of Norm-Minimizing Generalized Inverses}
\label{sec:computation}
For completeness, we briefly discuss the computation of the generalized inverses associated with various matrix norms. For simplicity, we only discuss \emph{real-valued} matrices.

General-purpose convex optimization packages such as CVX \cite{cvx,gb08} transform the problem into a semidefinite or second-order cone program, and then invoke the corresponding interior point solvers. This makes them quite slow, especially when the program cannot be reduced to a vector form.

For the kind of problems that we aim to solve, it is more appropriate to use methods such as the projected gradient method, or the alternating direction method of multipliers (ADMM) \cite{Parikh:2013}, also known as the Douglas-Rachford splitting method \cite{Eckstein:1992hp}. We focus on ADMM as it nicely handles non-smooth objectives.

\subsection{Norms that Reduce to the Vector Case}\label{sec:vectorcase}
In certain cases, it is possible to reduce the computation of a norm-minimizing generalized inverse to a collection of independent vector problems. This is the case with the columnwise mixed norms (and consequently entrywise norms).

Consider the minimization for $\ginv{\colpq{p}{q}}{\mA}$, with $q<\infty$. As $x \mapsto x^{\frac{1}{q}}$ is monotonically increasing over $\R^+$, it follows that
\begin{equation}
  \label{eq:decoupling}
  \argmin_{\mX \in \ginvset(\mA)} \left( \sum_j \norm{\vx_j}_p^q \right)^\frac{1}{q}
  = \argmin_{\mA \mX = \mI} \sum_j \norm{\vx_j}_p^q.
\end{equation}
On the right-hand side there is no interaction between the columns of $\mX$, because the constraint $\mA \mX = \mI$ can be separated into $m$ independent constraints $\mA \vx_j = \ve_j, \ j \in \set{1, 2, \ldots, m}$. We can therefore perform the optimization separately for every $\vx_j$, 
\begin{equation}
  \label{eq:per_column}
  \hat{\vx}_j = \argmin_{\mA \vx_j = \ve_j}~\norm{\vx_j}_p^q = \argmin_{\mA \vx_j =
  \ve_j}~\norm{\vx_j}_p.
\end{equation}
The above procedure also gives us \emph{a} minimizer for $q = \infty$.

This means that, in order to compute any $\ginv{\colpq{p}{q}}{\mA}$, we can use our favorite algorithm for finding the $\ell^p$-minimal solution of an underdetermined system of linear equations. The most interesting cases are  $p \in \set{1, 2, \infty}$. For $p = 2$, we of course get the MPP, and for the other two cases, we have efficient algorithms at our disposal \cite{Bach:2011uc,Duchi:2008en}.

\subsection{Alternating Direction Method of Multipliers (ADMM)}

The ADMM is  a method to solve minimization problems of the form
\begin{equation}
  \label{eq:splitting}
  \minimize \ f(x) + g(x),
\end{equation}
where in our case $x$ is a matrix $\mX$, and $g$ is a monotonically increasing function of a matrix norm $\norm{\mX}$. An attractive property of the method is that neither $f$ nor $g$ need be smooth.

\paragraph{The ADMM algorithm.} Generically, the ADMM updates are given as follows
\cite{Parikh:2013}:
 
\begin{equation}
\label{eq:admm}
\begin{aligned} 
  x^{k+1} &\bydef \prox_{\lambda f} (z^k - u^k) \\
  z^{k+1} &\bydef \prox_{\lambda g} (x^{k+1} + u^k) \\
  u^{k+1} &\bydef u^k + x^{k+1} - z^{k+1}.
\end{aligned}
\end{equation}
The algorithm relies on the iterative computation of proximal operators,
\[
\prox_{\varphi}(y) = \arg\min_{z} \tfrac{1}{2} \norm{y-x}_{2}^{2}+\varphi(x),
\]
which can be seen as a generalization of projections onto convex sets. Indeed, when $\varphi$ is an indicator function of a convex set, the proximal mapping is a projection.

\paragraph{Linearized ADMM.} 
Another attractive aspect of ADMM is that it warrants an immediate extension from ginv to pginv via the so called \emph{linearized ADMM}, which addresses optimization problems of the form

\begin{equation}
  \label{eq:splittinganalysis}
  \minimize \ f(x) + g\big(A(x)\big),
\end{equation}
with $A$ being a linear operator. The corresponding update rules, which involve its adjoint operator $A^{\star}$, are \cite{Parikh:2013}:

\begin{equation}
\label{eq:linearized_admm}
\begin{aligned} 
  x^{k+1} &\bydef \prox_{\lambda f} \big(x^k - (\mu/\lambda) A^\star\big(A (x^k) - z^k + u^k\big)\big) \\
  z^{k+1} &\bydef \prox_{\lambda g} \big(A (x^{k+1}) + u^k\big) \\
  u^{k+1} &\bydef u^k + A(x^{k+1}) - z^{k+1},
\end{aligned}
\end{equation}
where $\lambda$ and $\mu$ satisfy $0 < \mu \leq \lambda / \normlplq{2}{2}{A}^2$. Thus we may use almost the same update rules to optimize for $\norm{\mX}$ and for $\norm{\mX \mA}$, at a disadvantage of having to tune an additional parameter $\mu$.

As far as convergence goes, it can be shown that under mild assumptions, and with a proper choice of $\lambda$, ADMM converges to the optimal value. The convergence rate is in general sublinear; nevertheless, as Boyd notes in \cite{Boyd:2011tj}, ``Simple examples show that ADMM can be very slow to converge to high accuracy. However, it is often the case that ADMM converges to modest accuracy---sufficient for many applications---within a few tens of iterations.''
 
\subsection{ADMM for Norm-Minimizing Generalized Inverses} 
\label{sub:admm_for_norm_minimizing_generalized_inverses}

In order to use ADMM, we have to transform the constrained optimization program

\begin{equation}
\label{eq:constrained}
\begin{aligned}
& \underset{\mX}{\text{minimize}}
& & \norm{\mX} \\
& \text{subject to}
& &   \mX \in \ginvset(\mA)
\end{aligned}
\end{equation}
into an unconstrained program of the form \eqref{eq:splitting}. This is achieved by using indicator functions, defined as

\begin{equation}
  \ind{\mathcal{S}}(x) \bydef 
  \begin{cases}
  0, \quad x \in \mathcal{S} \\
  +\infty, \quad x \notin \mathcal{S}.
  \end{cases}
\end{equation}

The indicator function rules out any $x$ that does not belong to the argument set $\mathcal{S}$; in our case, $\mathcal{S}$ is the affine space $\ginvset(\mA)$. Using these notations, we can rewrite \eqref{eq:constrained} as an unconstrained program for computing $\ginv{\norm{\,\cdot\,}}{\mA}$:
\begin{equation}
\label{eq:unconstrained_ginv}
\begin{aligned}
& \underset{\mX}{\text{minimize}}
& & \ind{\ginvset(\mA)}(\mX) + h(\norm{\mX}),
\end{aligned}
\end{equation}
where $h$ is a monotonically increasing function on $\R^+$ (typically a power). Using similar reasoning, we can write the unconstrained program for computing $\pginv{\norm{\,\cdot\,}}{\mA}$ as
\begin{equation}
\label{eq:unconstrained_pginv}
\begin{aligned}
& \underset{\mX}{\text{minimize}}
& & \ind{\ginvset(\mA)}(\mX) + h(\norm{\mX \mA}).
\end{aligned}
\end{equation}
Next, we need to compute the proximal mappings for $h(\norm{\, \cdot \,})$ and for the indicator function.


\subsubsection{Proximal operator of the indicator function $\ind{\mathcal{\ginvset(\mA)}}(\mX)$}

The proximal operator of an indicator function of a convex set is simply the
Euclidean projection onto that set,

\[
  \prox_{\ind{\mathcal{S}}}(v) = \proj_{\mathcal{S}}(v) = \argmin_{x \in \mathcal{S}}~\norm{x - v}_2^2.
\]

In our case, $\mathcal{S} = \ginvset(\mA)$, which is an affine subspace

\[
  \ginvset(\mA) = \set{ \mA^\dagger + \mN \mZ \ : \ \mZ \in \R^{(n-m) \times
  (n-m)}},
\]
where we assumed that $\mA$ has full rank, and the columns of $\mN \in \R^{n \times (n-m)}$ form an orthonormal basis for the nullspace of $\mA$. To project $\mV$ orthogonally on $\ginvset(\mA)$, we can translate it and project it on the parallel linear subspace, and then translate the result back into the affine subspace:

\begin{equation}
  \prox_{\ind{\mathcal{\ginvset(\mA)}}}(\mV) = \proj_{\ginvset(\mA)}(\mV) = \mA^\dagger + \mN \mN^\T (\mV - \mA^\dagger) = \mA^\dagger + \mN \mN^\T \mV.
\end{equation}

\subsubsection{Proximal operators of some matrix norms}

Proximal operators associated with matrix norms are typically more involved than projections onto affine spaces. In what follows, we discuss the proximal operators for some mixed and induced norms.

An important ingredient is a useful expression for the proximal operator of any norm in terms of a projection onto a dual norm ball. We have that for any
scalar $\lambda > 0$,
\begin{equation}\label{eq:ProxIsProj}
  \prox_{\lambda \norm{\cdot}}(\mV) = \mV - \lambda \proj_{\set{\mX:\norm{\mX}_\star \leq 1}}(\mV),
\end{equation}
with $\norm{\, \cdot \,}_{\star}$ the dual norm to $\norm{\, \cdot \,}$. Thus computing the proximal operator of a norm amounts to projecting onto a norm ball of the dual norm. This means that we can compute the proximal operator efficiently if we can project efficiently, and vice-versa.

\paragraph{Mixed norms (columnwise and rowwise).} 
Even though we can compute $\ginv{\colpq{p}{q}}{\mA}$ by solving a series of vector problems, it is useful to rely on a common framework for the computations in terms of proximal operators; this will help later when dealing with $\pginv{\cdot}$.

Instead of computing the proximal mapping for $\normcolpq{p}{q}{\, \cdot \,}$, we compute it for $\normcolpq{p}{q}{\, \cdot \,}^q$ as it yields the same generalized inverse. Similarly as in Section~\ref{sec:vectorcase}, we have that
\begin{equation}
\begin{aligned}
  \prox_{\lambda\normcolpq{p}{q}{\, \cdot \,}^q}(\mV) 
  &= \argmin_{\mX} \sum_j \lambda \norm{\vx_j}_p^q + \tfrac{1}{2} \norm{\mV - \mX}_F^2 \\
  &= \argmin_{\mX} \sum_j \left( \lambda \norm{\vx_j}_p^q + \tfrac{1}{2} \norm{\vv_j - \vx_j}_2^2 \right)\\
  & = \left(\prox_{\lambda \norm{\cdot}_{p}^{q}}(\vv_{j})\right)_{j}
\end{aligned}
\end{equation}
where, unlike in Section~\ref{sec:vectorcase}, special care must be taken for $q = \infty$ (see Lemma~\ref{le:DualNorms} below).

Conveniently, the exact same logic holds for rowwise mixed norms: their proximal mappings can be constructed from vector proximal mappings by splitting into a collection of vector problems as follows:

\begin{equation}
\begin{aligned}
  \prox_{\lambda \normrowpq{p}{q}{\, \cdot \,}^q}(\mV) 
  &= \argmin_{\mX} \sum_i \lambda \norm{\vx^i}_p^q + \tfrac{1}{2} \norm{\mV - \mX}_F^2 \\
  &= \argmin_{\mX} \sum_i \left( \norm{\vx^i}_p^q + \tfrac{1}{2} \norm{\vv^i - \vx^i}_2^2 \right)\\
  & = \left(\prox_{\lambda \norm{\cdot}_{p}^{q}}((\vv^{i})^{\T})\right)_{i}^{\T}
\end{aligned}
\end{equation}

Interestingly, even though the computation of the generalized inverses corresponding to rowwise mixed norms does not decouple over rows or columns, we can decouple the computation of the corresponding proximal mappings as long as $q < \infty$. This is possible because the minimization for the latter is unconstrained.

Now the task is to find efficient algorithms to compute the proximal mappings for powers of vector $\ell^p$-norms, $\norm{\, \cdot \, }_p^q$. We list some known results for the most interesting (and the simplest) case of $q=1$, and $p \in \set{1, 2, \infty}$.

\begin{enumerate}
\item When $p = 1$ we have the so-called soft thresholding operator,
\[
  \prox_{\lambda \norm{\, \cdot \, }_1}(\vv) = \sign(\vv) \odot (\abs{\vv} -
  \lambda \vec{1})_+.
\] 
where $|\va|$, $\va \odot \vb$, $(\va)_{+}$ denote entrywise absolute value, multiplication, positive part, and $\vec{1}$ has all entries equal to one.
\item For $p=2$ we have that 
\[
  \prox_{\lambda \norm{\, \cdot \, }_2}(\vv) = \max\set{1 - \lambda / \norm{\vv}_2,\ 0} \vv.
\] 
The case of $p=2$ is interesting for rowwise mixed norms; for the columnwise norms, we simply recover the MPP.

\item Finally, for $p = \infty$, it is convenient to exploit the relationship with the projection operator \eqref{eq:ProxIsProj}. The dual norm to $\ell^\infty$-norm is the $\ell^1$-norm. We can project on its unit norm-ball as follows:

\[
  \proj_{\set{\vz : \norm{\vz}_1 \leq 1}}(\vv) = \sign(\vv) \odot (\abs{\vv} - \lambda \vec{1})_+,
\]
where
\[
  \lambda = 
  \begin{cases}
    0 & \norm{\vv}_1 \leq 1 \\
    \text{solution of}~\sum_i \max(\abs{x_i} - \lambda, 0) = 1 & \text{otherwise}.
  \end{cases}
\]
This can be computed in linear time \cite{Liu:2009if} and the proximal mapping is then given by \eqref{eq:ProxIsProj}.
\end{enumerate}
In all cases, to obtain the corresponding matrix proximal mapping, we simply concatenate the columns (or rows) obtained by vector proximal mappings.

\paragraph{Norms that do not (Easily) Reduce to the Vector Case}

The analysis of the proximal mappings in the previous section fails when $q = \infty$. On the other hand, computing the poor man's $\ell^1$-minimization generalized inverse $\ginv{\lplq{1}{1}}{\mA}$ using ADMM requires us to find the proximal operator for $\normlplq{1}{1}{\, \cdot \,} = \normcolpq{1}{\infty}{\, \cdot \,}$. Note that we could compute $\ginv{\colpq{1}{\infty}}{\mA}$ directly by decoupling, but we cannot do it for the corresponding $\pginv{}{}$ nor for the rowwise norms. We can again derive the corresponding proximal mapping using \eqref{eq:ProxIsProj}. The following lemma will prove useful:

\begin{lemma}\label{le:DualNorms}
  The dual norm of $\normlplq{1}{q}{\cdot} = \normcolpq{q}{\infty}{\cdot}$ is 
  \[
  \normlplq{1}{q}{\cdot}^{\H} = \normcolpq{q^\H}{1}{\cdot},
  \]
  and the dual norm of $\normlplq{p}{\infty}{\cdot} = \normrowpq{p^{\H}}{\infty}{\cdot}$ is 
  \[
  \normlplq{p}{\infty}{\cdot}^{\H} =  \normrowpq{p}{1}{\cdot},
  \]
  where $1/p + 1/p^\H = 1/q + 1/q^\H = 1$.
\end{lemma}

\begin{proof}
The first result is a direct consequence of the characterization~\eqref{eq:l1lq_characterization} of $\lplq{1}{q}$ norms in terms of columnwise mixed norms $\colpq{q}{\infty}$, and of \cite[Lemma~3]{Sra:2012ba}. The second result follows from the characterization~\eqref{eq:lqlinf_characterization} of $\lplq{p}{\infty}$ norms in terms of $\rowpq{p^{\H}}{\infty}$ and, again, of \cite[Lemma~3]{Sra:2012ba}.
\end{proof}

Combining~\eqref{eq:ProxIsProj} with Lemma~\ref{le:DualNorms} shows that for $\normlplq{1}{q}{\,\cdot\,}$, where $q \in \set{1,2,\infty}$, computing the proximal operator means projecting onto the ball of the $\normcolpq{q^\H}{1}{\,\cdot\,}$ norm. The good news is that these projections can be computed efficiently. We have already seen the cases when $q^* \in \set{1, 2}$ in the previous section (corresponding to $p \in \set{2, \infty})$.

For $q^* = \infty$ we can use the algorithm of Quattoni \emph{et al.} \cite{Quattoni:2009if} that computes the projection on the $\normcolpq{\infty}{1}{\, \cdot \,}$ (or equivalently $\normrowpq{\infty}{1}{\, \cdot \,}$) norm ball in time $O(mn \log (mn))$. Even for a general $q$, the projection can be computed efficiently, but the algorithm becomes more involved \cite{Liu:2010vp,Wang:2013tz}.\footnote{Our notation differs from the notation used in \cite{Liu:2010vp,Wang:2013tz,Quattoni:2009if}; the roles of $p$ and $q$ are reversed.} In summary, we can efficiently compute the proximal mappings for induced norms $\normlplq{1}{q}{\cdot}$, and equivalently for induced norms $\normlplq{p}{\infty}{\, \cdot \,}$, because these read as proximal mappings for certain mixed norms.

\subsection{Generic ADMM for matrix norm minimization} 

Given a norm $\norm{\cdot}$, we can now summarize the ADMM update rules for computing $\ginv{\norm{\cdot}}{\mA}$ and $\pginv{\norm{\cdot}}{\mA}$, assuming that $\prox_{\lambda \norm{\cdot}_{}^{q}}(\cdot)$ can be computed efficiently for some $0<q<\infty$.

\subsubsection{ADMM for the computation of $\ginv{}{\cdot}$}
To compute $\ginv{}{\cdot}$, we simply solve \eqref{eq:unconstrained_ginv} with $h(t) = \lambda t^{q}$ by running the iterates
\begin{equation}
\label{eq:admm_for_ginv}
\begin{aligned}
    \mX^{k+1} \ &\bydef \ \mA^\dagger + \mN \mN^\T (\mX^{k} - \mU^k) \\
    \mZ^{k+1} \ &\bydef \ \prox_{\lambda \norm{\cdot}^{q}}(\mX^{k+1} + \mU^k) \\
    \mU^{k+1} \ &\bydef \ \mU^k + \mX^{k+1} - \mZ^{k+1}.
\end{aligned}  
\end{equation}
While there are a number of references that study the choice of the
regularization parameter $\lambda$ for particular $f$ and $g$, this choice
still remains somewhat of a black art. Discussion of this topic falls out of
the scope of the present contribution. In our implementations\footnote{Available online at \url{https://github.com/doksa/altginv}.} we used the values of $\lambda$ for which the algorithm was
empirically verified to converge.

\subsubsection{Linearized ADMM for the computation of $\pginv{}{\cdot}$}

Even for entrywise and columnwise mixed norms, things get more complicated when instead of $\ginv{\colpq{p}{q}}{\mA}$ we want to compute $\pginv{\colpq{p}{q}}{\mA}$. This is because the objective $\normcolpq{p}{q}{\mX \mA}$ now mixes the columns of $\mX$ so that they cannot be untangled. A similar issue arises when trying to compute $\pginv{\rowpq{p}{q}}{\mA}$. This issue is elegantly addressed by the linearized ADMM, and by using the proximal mappings described in the previous section.

The linearized ADMM allows us to easily adapt the updates \eqref{eq:admm_for_ginv} for norms on the matrix $\mX \mA$, without computing the new proximal operator. To compute $\pginv{}{\cdot}$, we express \eqref{eq:unconstrained_pginv} with $h(t) = \lambda t^{q}$. This has the form
\begin{equation}
  \minimize \ f(\mX) \ + \ g(\mX \mA).
\end{equation}
where $g(\mV) = h(\norm{\mV})$. It is easy to verify that the adjoint of $A: \mX \in \R^{n \times m} \to \mX\mA \in \R^{n \times n}$ is $A^{\star}: \mV \in \R^{n \times n} \to \mV\mA^{\T} \in \R^{n \times m}$, and that $\normlplq{2}{2}{A} = \normlplq{2}{2}{\mA}$ so that the updates for $\pginv{\norm{\cdot}}{\mA}$ are given as
\begin{equation}
\begin{aligned}
    \mX^{k+1} \ &\bydef \ \mA^\dagger + \mN \mN^\T\left[ \mX^k - (\mu/\lambda) (\mX^k \mA - \mZ^k + \mU^k)\mA^\T \right] \\
    \mZ^{k+1} \ &\bydef \ \prox_{\lambda \norm{\cdot}^{q}}(\mX^{k+1} \mA + \mU^k) \\
    \mU^{k+1} \ &\bydef \ \mU^k + \mX^{k+1}\mA - \mZ^{k+1},
\end{aligned}  
\end{equation}
where $0 < \mu \leq \lambda / \normlplq{2}{2}{\mA}^2$.


\section{Conclusion} 
\label{sec:conclusion}

We presented a collection of new results on generalized matrix inverses which minimize various matrix norms. Our study is motivated by the fact that the Moore-Penrose pseudoinverse minimizes e.g. the Frobenius norm, by the intuition that minimizing the entrywise $\ell_1$ norm of the inverse matrix should give us sparse pseudoinverses, and by the fact that poor man's $\ell^p$ minimization---a linear procedure which minimizes the worst-case $\ell^p$ norm blowup---is achieved by a generalized inverse which minimizes the induced $\ell^1 \to \ell^1$ norm of the associated projection matrix. 

Most of the presented findings address the relation in which various norms and matrices stand with respect to the MPP. In this regard, a number of findings make our work appear Sisyphean since for various norms and matrices we merely reestablish the optimality of the MPP. We could summarize this in a maxim ``When in doubt, use the MPP'', which most practitioners will not find surprising.

On the other hand, we identify classes of matrix norms for which the above statement does not hold, and whose minimization leads to matrices with very different properties, potentially useful in applications. Perhaps the most interesting such generalized inverse---the sparse pseudoinverse---is studied in Part II of this two-paper series.

Future work related to the results presented in this Part I involves extensions of our results to rank deficient matrices, further study of $\pginv{}{\cdot}$ operators, and filling several ``holes'' in the results as we could not answer all the posed questions for all combinations of norms and matrices (see Table~\ref{tab:norm-list}), in particular for induced norms.



\section{Acknowledgments}

The authors would like to thank Mihailo Kolund\v{z}ija, Miki Elad, Jakob Lemvig, and Martin Vetterli for the discussions and input in preparing this manuscript, and Laurent Jacques for pointing out related work. This work was supported in part by the European Research Council, PLEASE project (ERC-StG-2011-277906).


\section*{Appendices}
\appendix

\section{Proofs of Formal Statements}

\subsection{Proof of Corollary~\ref{cor:classicalMPP}} 
\label{sub:proof_of_classicalMPP}

We begin by showing that the MPP is {\em a} minimizer of the considered norms.
The result for $\norm{\cdot}_{S_{p}}$ with $1 \leq p \leq \infty$, for the
induced norm $\norm{\cdot}_{\lplq{2}{2}}$ and columnwise mixed norm
$\norm{\cdot}_{\colpq{2}{2}}$ follow from their unitarily invariance and
Theorem~\ref{th:Zietak}. In contrast, the norms $\norm{\cdot \mA}_{S_{p}}$
generally fail to be unitarily invariant. Similarly, for $p \neq 2$ and $q
\neq 2$, induced norms and columnwise mixed norms are {\em not} unitarily
invariant, hence the results do not directly follow from
Theorem~\ref{th:Zietak}. Instead, the reader can easily check that all
considered norms are left unitarily invariant, hence the MPP is {\em a}
minimizer by Theorem~\ref{th:LUINorms} and Corollary~\ref{cor:PLUINorms}.

{\em Uniqueness cases for Schatten norms and columnwise mixed norms.} Since
Schatten norms with $1 \leq p<\infty$ are fully unitarily invariant and
associated to a strictly monotonic symmetric gauge function, the uniqueness
result of Theorem~\ref{th:Zietak} applies. To establish uniqueness with the
norms $\norm{\cdot \mA}_{S_{p}}$, $1 \leq p<\infty$ we exploit the following
useful Lemma (see, e.g., \cite[Lemma 7]{Zietak:1997ut} and references
therein for a proof).
\begin{lemma}
  \label{lem:uinv-block}
  Let $\mX, \mY \in \C^{m \times n}$ be given as
  \begin{equation}
    \mX =
    \begin{bmatrix}
     \mK_1 & \mat{0} \\
      \mat{0} & \mat{0}
   \end{bmatrix},
    \quad \quad 
    \mY = 
    \begin{bmatrix}
      \mK_1 & \mK_2 \\
      \mK_3 & \mK_4
      \end{bmatrix}
  \end{equation}
  where the block $\mK_{1}$ is of size $r$. Then we have that $\sigma_j(\mX) \leq \sigma_j(\mY)$ for $1 \leq j \leq r$, and for any unitarily invariant norm  $\norm{\cdot}_{\phi}$ associated to a symmetric gauge $\phi$, we have $\norm{\mX}_{\phi} \leq
  \norm{\mY}_{\phi}$. When $\phi$ is strictly monotonic,
   $\norm{\mX}_{\phi} = \norm{\mY}_{\phi}$ if and only if $\mK_2,
  \mK_3$ and $\mK_4$ are zero blocks.
\end{lemma}
Considering $\mX \in \ginvset(\mA)$ and its representation as given
in~\eqref{eq:GinvsetMform}, using the unitary invariance of the Schatten norm,
we have
\begin{equation}
\label{eq:unique_schatten_pginv}
\norm{\mX\mA}_{S_{p}} = \norm{\mM\mSigma}_{S_{p}}
=
\norm{    
\begin{bmatrix}
    \mI_{m} & \mat{0} \\
    \mS \mSigma_{\square} \ & \mat{0}
\end{bmatrix}
}_{S_{p}} >
\norm{    
\begin{bmatrix}
      \mI_{m} & \mat{0} \\
      \mat{0} \ & \mat{0}
\end{bmatrix}
}_{S_{p}} 
=
\norm{\mA^{\dagger}\mA}_{S_{p}}
\end{equation}
as soon as $\mS \mSigma_{\square} \neq \mat{0}$, that is to say whenever $\mS
\neq \mat{0}$. For both types of columnwise mixed norms with $1 \leq q < \infty$, the
strictness of the inequality in Lemma~\ref{lem:luinv-block} when $\mK_{2} \neq
\mat{0}$ is easy to check, and we can apply the uniqueness result of
Corollary~\ref{cor:PLUINorms}. 

{\em Uniqueness for $\pginv{S_\infty}{\mA}$.} To prove uniqueness for $\pginv{S_\infty}{\mA}$, consider again equation
\eqref{eq:unique_schatten_pginv} with $p=\infty$. Denote $\mZ \bydef \mS
\mSigma_\square \neq \mat{0}$, and let $\vz^*$ be its first nonzero row. 
We bound the largest singular
value $\sigma_{1}$ of the matrix $\mY$ defined as follows
\[
    \mY \bydef
    \begin{bmatrix}
       \mI_m & \mat{0} \\
       \vz^* & 0
    \end{bmatrix}.
\]
By the variational characterization of eigenvalues of a Hermitian matrix, we have
\begin{equation}
    \sigma_1^2 = \max_{\norm{\vx}_2 = 1} \vx^* \mY \mY^\H \vx
    = \max_{\norm{\vx}_2 = 1}
    \vx^*
    \begin{bmatrix}
       \mI_m & \vz \\
       \vz^* & \norm{\vz}_2^2
    \end{bmatrix}
    \vx.
\end{equation}
Partitioning $\vx$ as $\vx^{*} = (\wt{\vx}^{*}, \  x_n^{*})$ with $x_{n} \in \mathbb{C}$ we write
the maximization as 
\begin{equation}
    \max_{\norm{\wt{\vx}}_2^2 + |x_n|^2 = 1} \norm{\wt{\vx}}_2^2 + 2 \Re (x_n^{*} \vz^* \wt{\vx}) + |x_n|^2 \norm{\vz}_2^2.
\end{equation}
Since $\vz \neq 0$ by assumption, it must have at least one non-zero entry.
Let $i$ be the index of a non-zero entry, and restrict $\wt{\vx}$ to the form
$\wt{\vx} = \alpha \ve_i$, $\alpha \in \mathbb{C}$. We have that 
\begin{equation}
    \max_{\norm{\wt{\vx}}_2^2 + |x_n|^2 = 1} \norm{\wh{\vx}}_2^2 + 2 \Re(x_n^{*} \vz^* \wt{\vx}) + |x_n|^2 \norm{\vz}_2^2
    \geq
    \max_{|\alpha| \leq 1} |\alpha|^2 + 2 \sqrt{1-|\alpha|^2} \cdot |z_i | \cdot |\alpha|  + (1-|\alpha|^2) |z_i|^2,
\end{equation}
where we used the fact that maximization over a smaller set can only diminish
the optimum, as well as $\norm{\vz}_2^2 \geq |z_i|^2$. Straightforward calculus
shows that the maximum of the last expression is $1 + |z_i|^2$, which is
strictly larger than $1$. Since $\mY \mY^\H$ is a principal minor of
\[
\begin{bmatrix}    
    \mI_{m} & \mat{0} \\
    \mS \mSigma_{\square} \ & \mat{0}
\end{bmatrix}
\begin{bmatrix}    
    \mI_{m} & \mat{0} \\
    \mS \mSigma_{\square} \ & \mat{0}
\end{bmatrix}^\H
\]
it is easy to see that
\[
\norm{
\begin{bmatrix}    
    \mI_{m} & \mat{0} \\
    \mS \mSigma_{\square} \ & \mat{0}
\end{bmatrix}
}_{S_\infty}
 \geq
 \sqrt{\max_{\norm{\vx}_2 = 1} \vx^* \mY \mY^\H \vx}
 = 
\sigma_1 \geq \sqrt{1 + |z_i|^2} \stackrel{\text{strictly}}{>} 1 = \norm{\mA^\dagger \mA}_{S_\infty}.
\]

{\em Cases of possible lack of uniqueness.} The construction of $\mA^{\ddagger} \neq \mA^{\dagger}$ in Example~\ref{ex:MPPnonUniqueOpNorm} provides a matrix $\mA$ for which $\ginv{S_{\infty}}{\mA} \supset \{\mA^{\dagger},\mA^{\ddagger}\}$ is not reduced to the MPP. 

A counterexample for $\ginv{\colpq{2}{\infty}}{\cdot}$ is as follows: consider $0< \eta <  1$ and the following matrix, 
\begin{equation}
\label{eq:CounterExampleCorollary2}
  \mA =
  \begin{bmatrix}
    \eta^{-1} & 0 & 0 \\
    0 & 1 & 0
  \end{bmatrix}.
\end{equation}
Its MPP is simply $\mA^{\dagger} = \begin{bmatrix} \eta & 0 & 0 \\ 0 & 1 & 0
\end{bmatrix}^\T$, and we have that $\normcolpq{2}{\infty}{\mA^{\dagger}} =
1$. One class of generalized inverses is given as
\begin{equation}
  \mA^{\ddagger}(\alpha) = 
  \begin{bmatrix}
    \eta & 0 \\
    0 & 1 \\
    \alpha & 0
  \end{bmatrix}.
\end{equation} 
Clearly for all $\alpha$ with $\abs{\alpha} \leq \sqrt{1-\eta^2}$, $\normcolpq{2}{\infty}{\mA^{\ddagger}(\alpha)} = \normcolpq{2}{\infty}{\mA^{\dagger}} = 1$, hence $\ginv{\colpq{2}{\infty}}{\cdot}$ is not a singleton.
To construct a counterexample for $\pginv{\colpq{2}{\infty}}{\cdot}$, consider a
full rank matrix $\mA$ with the SVD $\mA = \mU \mSigma \mV^\H$. Then all the generalized inverses $\mA^\ddagger$ are given as
\eqref{eq:GinvsetMform}, so that
\begin{equation}
    \normcolpq{2}{\infty}{\mA^{\ddagger} \mA} = \normcolpq{2}{\infty}{ 
    \begin{bmatrix}
        \mV_1^\H \\
        \mS \mSigma \mV_1^\H
    \end{bmatrix}},
\end{equation}
where we applied the left unitary invariance and we partitioned $\mV = [\mV_1 \ \mV_2]$. Clearly, setting $\mS = \mat{0}$ gives the MPP and
optimizes the norm. To construct a counterexample, note that generically
$\mV_1^\H$ will have columns with different 2-norms. Let $\vv$ be its column
with the largest 2-norm; we choose a matrix $\mP$ so that $\mP \mSigma \vv = \vec{0}$ while $\mP \mSigma \mV_{1}^{\H} \neq \mat{0}$.
Then there exists $\varepsilon_0$ so that for any $\varepsilon \leq \varepsilon_0$,
\begin{equation}
    \normcolpq{2}{\infty}{
    \begin{bmatrix}
        \mV_1^\H \\
        \varepsilon \mP \mSigma \mV_1^\H
    \end{bmatrix}}
    =
    \normcolpq{2}{\infty}{
    \begin{bmatrix}
        \mV_1^\H \\
        \mat{0}
    \end{bmatrix}}
    = \normcolpq{2}{\infty}{\mA^\dagger \mA}.
\end{equation}

We can reuse counterexample~\eqref{eq:CounterExampleCorollary2} to show the lack of uniqueness for the induced
norms (note that we already have two counterexamples: for $p = 1$ as this
gives back the columnwise norm, and for $p=2$ as this gives back the spectral
norm which is the Schatten infinity norm). Let us look at
$\normlplq{p}{2}{\mA^\ddagger(\alpha)} = \max_{\vu \in \mathbb{R}^{2}: \norm{\vu}_p = 1} \norm{\mA^{\ddagger}(\alpha) \vu}_2$. We can
write this as (after squaring the 2-norm)
\begin{equation}
    \label{eq:counterexample_induced_ellipse}
    \max_{u_1^p + u_2^p = 1} (\eta^2 + \alpha^2) u_1^2 + u_2^2.
\end{equation}
The optimization problem~\eqref{eq:counterexample_induced_ellipse} is depicted geometrically in Figure \ref{fig:geometric_counterexample}: consider the ellipse with equation
\begin{equation}
    (\eta^2 + \alpha^2)u_1^2 + u_2^2 = R^2,
\end{equation}
We are searching for the largest $R$ so that there exists a point on this ellipse with unit $\ell^p$ norm. In other words, we are squeezing the ellipse until it touches the $\ell^p$ norm ball. 
\begin{figure}[t]
   \centering
   \includegraphics{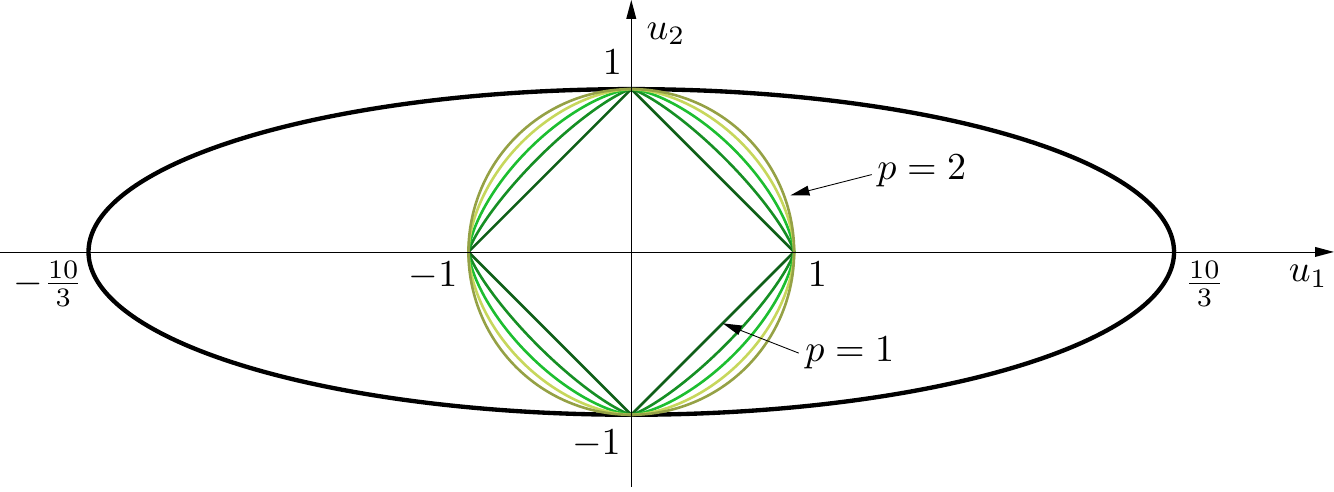}
   \caption{Geometry of the optimization problem~\eqref{eq:counterexample_induced_ellipse}. Unit norm balls $\abs{u_1}^p + \abs{u_2}^p = 1$ are shown for $p \in \set{1, 1.25, 1.5, 1.75, 2}$.}
   \label{fig:geometric_counterexample}
\end{figure}
The semi-axes of the ellipse are $R / \sqrt{\eta^2 + \alpha^2}$ and $R$, so that when $\eta$ and $\alpha$ are both small, the ellipses are elongated along the horizontal axis and the intersection between the squeezed ellipse and the $\ell^{p}$ ball will be close to the vertical axis. In fact, for $p \leq 2$ , one can check%
\footnote{This is no longer the case for $p > 2$ for the $\ell^{p}$ ball is ``too smooth'' at the point $[0,1]$: around this point on the ball we have $1-u_{2} = c |u_{1}|^{p} (1+o(u_{1})) = o(u_{1}^{2})$.}
that there exists $0< \beta < 1$ such that, whenever $\eta^2 + \alpha^2 \leq \beta^{2}$, the squeezed ellipse touches the $\ell^{p}$ ball only at the points $[0, \pm 1]^\T$ (as seen on Figure \ref{fig:geometric_counterexample}). That is, the maximum is achieved for $R = 1$. The value of the cost function at this maximum is
\begin{equation}
    (\eta^2 + \alpha^2) 0^2 + (\pm 1)^2  = 1
\end{equation}
Therefore, choosing $\eta < \beta$, the value of $\normlplq{p}{2}{\mA^\ddagger(\alpha)}$ is constant for all $\alpha^2 \leq \beta^{2}-\eta^2$, showing that there are many generalized inverses yielding the same $\lplq{p}{2}$ norm as the MPP, which we know achieves the optimum.

\subsection{Proof of Theorem~\ref{thm:not_always_mpp}} 
\label{sub:proof_of_theorem_not_always_mpp}

We will use a helper lemma about gradients and subdifferentials of matrix norms:

\begin{lemma}\label{le:subgradients0}
For a real-valued function of a complex matrix $\mZ = (z_{k\ell}) \in \C^{m \times n}$, $f(\mZ)$, denote
\[
    \nabla_\mZ f(\mZ) \bydef 
        \left[ \frac{\partial}{\partial x_{k\ell}} f(\mX+j\mY)
             + j \frac{\partial}{\partial y_{k\ell}}  f(\mX+j\mY) \right]_{k\ell}
\]
with $j$ the imaginary unit, $\mZ = \mX+j\mY$, $\mX = (x_{k\ell}) \in \R^{m \times n}$, $\mY = (y_{k\ell}) \in \R^{m \times n}$. Denote $|\cdot|^{r}$ the entrywise $r$-th power of the (complex) magnitude of a vector or matrix, $\odot$ the entrywise multiplication, $\sign(\cdot)$ the entrywise sign with $\sign(z) \bydef z/\abs{z}$ for nonzero $z \in \C$, and $\sign(0) \bydef \set{u \in \C, \abs{u} \leq 1}$. 
The chain rule yields:
\begin{itemize}
\item for $f(z) = \abs{z}$ the complex magnitude, we have for any nonzero scalar $z \in \C$
\[ 
\nabla_z \abs{z} = \frac{z}{\abs{z}} = \sign(z).
\]
At $z = 0$, the subdifferential of $f$ is $\partial_{z} f(0) = \sign(0)$.
\item for $f(z) = \abs{z}^{p}$ with $0<p<\infty$, we have for any nonzero scalar $z \in \C$
\[ 
\nabla_z \abs{z}^{p} = p \abs{z}^{p-1} \nabla_{z} \abs{z} = p \abs{z}^{p-1} \sign(z).
\]
At $z = 0$, this also holds for $p>1$, yielding $\nabla_{z} f(0) = 0$. See above for $p=1$.
\item for $f(\vz) = \norm{\vz}_p = \left(\sum_{k=1}^{m} \abs{z_k}^p\right)^{\tfrac{1}{p}}$ with $0<p<\infty$, we have for any vector with nonzero entries $\vz \in \C^{m}$
\[
\nabla_\vz \norm{\vz}_{p} = \tfrac{1}{p} (\norm{\vz}_p^p)^{\tfrac{1}{p}-1} \left[ \nabla_{z_k} \abs{z_k}^p \right]_k 
                = \norm{\vz}_p^{1-p} \left(\abs{\vz}^{p-1}\odot \sign(\vz)\right).
\]
When $\vz$ has some (but not all) zero entries: this also holds for $p>1$, and the entries of $\nabla_{\vz} f(\vz)$ corresponding to zero entries of $\vz$ are zero; for $p=1$, the above expression yields the subdifferential of $f$ at $\vz$, $\partial_{\vz} f(\vz) = \sign(\vz)$.
 \\
At $\vz=0$, the subdifferential of $f$ is $\partial_{\vz} f(\vec{0}) = \mathcal{B}_{p^{*}} = \set{\vu \in \C^{m}, \norm{\vu}_{p^{*}} \leq 1}$ with $\tfrac{1}{p^{*}}+\tfrac{1}{p}=1$.

\item for $f(\mZ) = \norm{\mZ}_{\colpq{p}{q}} = \left(\sum_{\ell=1}^{n} \norm{\vz_{\ell}}_{p}^{q}\right)^{\frac{1}{q}}$, $0< p,q<\infty$, we have for any matrix with nonzero entries $\mZ \in \C^{m \times n}$:
\begin{eqnarray*}
\nabla_{\mZ} \norm{\mZ}_{\colpq{p}{q}}
&=& \tfrac{1}{q} \left(\norm{\mZ}_{\colpq{p}{q}}^{q}\right)^{\tfrac{1}{q}-1} \left[ \nabla_{\vz_{\ell}} \norm{\vz_\ell}_{p}^{q}\right]_{\ell}
= \tfrac{1}{q} \norm{\mZ}_{\colpq{p}{q}}^{1-\tfrac{1}{q}} \left[ q \norm{\vz_{\ell}}_{p}^{q-1} \nabla_{\vz_{\ell}} \norm{\vz_{\ell}}_{p}\right]_{\ell}\\
&=&  \norm{\mZ}_{\colpq{p}{q}}^{1-\tfrac{1}{q}} \left[ \norm{\vz_{\ell}}_{p}^{q-1} \norm{\vz_{\ell}}_{p}^{1-p} (\abs{\vz_{\ell}}^{p-1} \odot \sign(\vz_{\ell}))\right]_{\ell}\\
&=&  \norm{\mZ}_{\colpq{p}{q}}^{1-\tfrac{1}{q}} \left( \abs{\mZ}^{p-1} \odot \sign(\mZ)\right) \mD_{\texttt{col}}^{q-p}
\end{eqnarray*}
with $\mD_{\texttt{col}} \bydef \diag(\norm{\vz_{\ell}}_{p})$.

When $\mZ$ has nonzero columns but some zero entries: this also holds for $p>1$, and the entries of $\nabla_{\mZ} \norm{\mZ}_{\colpq{p}{q}}$ corresponding to zero entries of $\mZ$ are zero; for $p=1$ the above expression yields the subdifferential of $f$ at $\mZ$, $\partial_{\mZ} f(\mZ) = \norm{\mZ}_{\colpq{1}{q}}^{1-\tfrac{1}{q}} \sign(\mZ) \mD_{\texttt{col}}^{q-1}$.

When $\mZ$ has some (but not all) zero columns, we get
\begin{itemize}
    \item For $p > 1$ and $q > 1$ the above extends, with zeros in the gradient entries (resp. columns) corresponding to zero entries (resp. columns) in $\mZ$.
    \item For $p > 1$ and $q = 1$, $f(\mZ) = \normcolpq{p}{1}{\mZ} = \sum_{\ell=1}^n \norm{\vz_\ell}_p$, and we get a subdifferential 

    \[
        \partial_{\mZ} \normcolpq{p}{1}{\mZ}
        = 
        \set{ \mV = [\vv_1, \ldots, \vv_n] \ \bigg| \ \vv_i = \nabla \norm{\vv_i}_p~\text{if}~\vz_i \neq \vec{0}, \ \vv_i \in \partial \norm{\vz_i}_p~\text{otherwise}}
    \]
    \item For $p = 1$ and $q \geq 1$ we again get the subdifferential $\partial_{\mZ} \normcolpq{1}{q}{\mZ} = \normcolpq{p}{q}{\mZ}^{1-\tfrac{1}{q}} \sign(\mZ) \mD_{\texttt{col}}^{q-p}$ as above. In particular, when $p = q = 1$, we have that $\partial_{\mZ} \normcolpq{1}{1}{\mZ} = \sign(\mZ)$.
\end{itemize}

At $\mZ = 0$ the subdifferential is again the unit ball of the dual norm (this follows from the definition of a subdifferential and the definition of a dual norm).

\item similarly for $f(\mZ) = \norm{\mZ}_{\rowpq{p}{q}} = \left(\sum_{k} \norm{\vz^{k}}_{p}^{q}\right)^{\frac{1}{q}}$,  $0< p,q<\infty$, we have for any matrix with nonzero entries $\mZ \in \C^{m \times n}$:
\begin{eqnarray*}
\nabla_{\mZ} \norm{\mZ}_{\rowpq{p}{q}}
&=&  \norm{\mZ}_{\rowpq{p}{q}}^{1-\tfrac{1}{q}} \mD_{\texttt{row}}^{q-p} \left( \abs{\mZ}^{p-1} \odot \sign(\mZ)\right) 
\end{eqnarray*}
with $\mD_{\texttt{row}} \bydef \diag(\norm{\vz^{k}}_{p})$.

The same extensions to $\mZ$ with some (but not all) zero entries (resp. zero rows) hold.
In particular for $p=1$ and $\mZ$ with nonzero rows, $\partial_{\mZ} f(\mZ) = \norm{\mZ}_{\rowpq{1}{q}}^{1-\tfrac{1}{q}} \mD_{\texttt{row}}^{q-1} \sign(\mZ)$.
\end{itemize}
The same results hold for the same functions and norms of {\em real-valued} vectors and matrices.
\end{lemma}

\begin{proof}[Proof of Theorem~\ref{thm:not_always_mpp}]
For simplicity we use $g(\mX)$ to denote the considered norm (respectively $\normcolpq{p}{q}{\mX}$, $\normcolpq{p}{q}{\mX\mA}$, $\normrowpq{p}{q}{\mX}$ or $\normrowpq{p}{q}{\mX\mA}$). Our goal is to build a matrix $\mA$ such that: for some $\mX \in \ginvset(\mA)$ we have $g(\mA^{\dagger}) > g(\mX)$.

For any $\mA$ with full column rank (a property that will be satisfied by the matrices we will construct), any $\mX \in \ginvset(\mA)$ can be written as $\mA^\dagger + \mN \mZ$, where $\mN \in \R^{n\times (n-m)}$ is a basis for the nullspace of $\mA$, and $\mZ \in \R^{(n-m)\times m}$. Using this representation and denoting $f(\mZ) \bydef g(\mA^{\dagger}+\mN \mZ)$, it is enough to show that the {\em subdifferential} of $f$ at $\mZ = \mat{0}$ does not contain $\mat{0}$, i.e., that
\[
\mat{0}_{(n-m)\times m} \notin \partial f(\mat{0}) = 
\set{\mF \in \R^{(n-m) \times m}:\ \forall \mZ,\ f(\mZ)-f(\mat{0}) \geq \langle \mF,\mZ\rangle_{F}}.
\]
which is in turn equivalent 
to $\mat{0}_{(n-m) \times m} \notin \mN^\T  \partial g(\mA^\dagger)$. Put differently, there is no $n \times m$ matrix $\mG \in \partial g(\mA^\dagger)$ such that its columns are all orthogonal to the null space  of $\mA$, $\ker(\mA) = \range(\mN)$, i.e., for any $\mG = [\vg_{1},\ldots,\vg_{m}] \in \partial g(\mA^\dagger)$, there is a column index $j \in \set{1, \ldots, m}$ such that $\vg_{j} \notin \range(\mN)^{\perp} = \range(\mA^\T) = \range(\mA^\dagger)$. 
At this point we have shown that, given a full rank matrix $\mA$, there exists $\mX \in \ginvset(\mA)$ such that $g(\mA^{\dagger}) > g(\mX)$ if, and only if,
  \begin{equation}
    \label{eq:gradient_condition}
    \mat{0}_{n \times m} \notin (\mI_{n} - \mA^\dagger \mA) \partial g(\mA^\dagger) .
  \end{equation}
Our new goal is thus to exhibit such a matrix $\mA$ for all the considered norms.
\begin{enumerate}
\item {\bf Case of $\ginv{}{\cdot}$ for columnwise norms.} Since $\ginv{\colpq{p}{q}}{\cdot} = \ginv{\colpq{p}{p}}{\cdot} = \ginv{p}{\cdot}$ for $q<\infty$, it is sufficient to consider $q=p$, i.e. $g(\mX) = \norm{\mX}_{p} = \left(\sum_{j=1}^{m} \norm{\vx}_{p}^{p}\right)^{1/p}$. 
By Lemma~\ref{le:subgradients0}, we need to consider 
\[
h^{\mathsf{ginv}}_{\colpq{p}{p}}(\mA) 
=  (\mI_{n}-\mA^{\dag}\mA) \left( \abs{\mA^{\dag}}^{p-1} \odot \sign(\mA^{\dag})\right) 
\]
and exhibit some full rank $\mA$ such that $h^{\mathsf{ginv}}_{\colpq{p}{p}}(\mA)  \ni \mat{0}$ if, and only if, $p = 2$. Since $\mA^{\dag}\mA$ is the orthogonal projection onto the span of $\mA^{\dag}$, it is sufficient to find  a full rank $\mB \in \R^{n \times m}$ with strictly positive entries so that $(\mI_{n}-\mP_{\mB}) \abs{\mB}^{p-1} = \mat{0}$ if, and only if, $p = 2$, with $\mP_{\mB} = \mB\mB^{\dag}$ the orthogonal projection onto the span of $\mB$ (the matrix $\mB$ plays the role of $\mA^{\dag}$, and $\mB^{\dag}$ the role of $\mA$). 

Let $\mC \in \R^{n \times m}$ be a matrix with $0/1$ entries and $\mB \bydef \mat{1}_{n} \mat{1}_{m}^{\T} + \mC$. All entries of $\mB$ are ones except those associated to nonzero entries in $\mC$, which are twos. Hence
\begin{equation}\label{eq:Bconstruction}
\abs{\mB}^{p-1} 
= \mat{1}_{n}\mat{1}_{m}^{\T}+(2^{p-1}-1) \mC
= \mat{1}_{n}\mat{1}_{m}^{\T}+(2^{p-1}-1) (\mB-\mat{1}_{n}\mat{1}_{m}^{\T})
= (2^{p-1}-1) \mB -(2^{p-1}-2)\mat{1}_{n}\mat{1}_{m}^{\T}
\end{equation}
and
\[
(\mI_{n}-\mP_{\mB}) \abs{\mB}^{p-1} = -(2^{p-1}-2) (\mI_{n}-\mP) \mat{1}_{n} \mat{1}_{m}^{\T}.
\]
Provided that $\mat{1}_{n}$ does not belong to the span of $\mB$, $(\mI_{n}-\mP) \mat{1}_{n} \mat{1}_{m}^{\T}$ is nonzero, and the above expression is zero if, and only if, $2^{p-1} = 2$, which is equivalent to $p=2$.

To conclude, we just need to show we can choose a binary $\mC$ so that: (i) $\mB$ is full rank; and (ii) its span does not contain the vector $\mat{1}_{n}$. We let the reader check that for any $m \geq 1$ this can be achieved with
\[
\mC_= \begin{bmatrix}\mI_{m}\\ \mat{0}_{(n-m) \times m}\end{bmatrix}.
\]
This shows the desired property with $\mA_{1} \bydef \mB^{\dag}$.

\item {\bf Case of $\ginv{}{\cdot}$ for rowwise norms.}

For $m=1$ with rowwise norms, we have for any $0<p,q\leq \infty$ and $\mX \in \R^{n \times 1}$: 
$\normrowpq{p}{q}{\mX} = \norm{\mX}_{q}$. Hence for any $0<p,q \leq \infty$ and any $\mA \in \R^{1 \times n}$, 
$\ginv{\rowpq{p}{q}}{\mA} = \ginv{q}{\mA}$. This yields the desired property with $\mA_{1}$ obtained above.

For $m \geq 2$, by Lemma~\ref{le:subgradients0}, what we seek is  $\mA$ so that 
\[
h^{\mathsf{ginv}}_{\rowpq{p}{q}}(\mA) 
=  (\mI_{n}-\mA^{\dag}\mA) \diag\left(\norm{(\mA^{\dag})^{i}}_{p}^{q-p}\right)_{i=1}^{n}\left( \abs{\mA^{\dag}}^{p-1} \odot \sign(\mA^{\dag})\right) \ni \mat{0}
\]
if, and only if, $(p,q)$ have particular values. As previously, it is sufficient to find a full rank  $\mB \in \R^{n \times m}$ with no zero entry so that, with $\mD_{\text{row}} = \diag(\norm{\vb^{i}}_{p})$,
\[
(\mI_{n}-\mP_{\mB}) \mD_{\text{row}}^{q-p} (\abs{\mB}^{p-1} \odot \sign(\mB)) = \mat{0}
\]
only for these values of $(p,q)$.

We consider $\mB \bydef \vec{1}_n \vec{1}_m^\T + \mC'$ with 
\[
    \mC' \bydef 
    \left[\begin{array}{c}
        \mI_{m}\\ 
        \hline
        -2,1,\mat{0}_{m-2}^{\T}\\
        \hline
        \mat{0}_{(n-1-m) \times (m-1)}; \mat{1}_{(n-1-m) \times 1}
    \end{array}\right]
\]
so that $b_{(m+1)1} = -1$. In each line of $\abs{\mB}$, one entry is a $2$ and the remaining $m-1$ entries are ones, hence for  $0<p,q<\infty$ we have $\mD_{\text{row}}^{q-p} = C_{p,q} \mI_n$ and
\[
(\mI_n - \mP_\mB) \mD^{q-p}_\text{row} \left[ \abs{\mB}^{p-1} \odot \sign(\mB) \right] = C_{p,q} (\mI_n - \mP_\mB) \left[ \abs{\mB}^{p-1} \odot \sign(\mB) \right].
\]
The first column of $\left[ \abs{\mB}^{p-1} \odot \sign(\mB) \right]$ is $\vy \bydef \left[ 2^{p-1}, \vec{1}_{m-1},-1, \vec{1}_{n-1-m} \right]^\T$. To conclude we show that this lives in the span of $\mB$ if, and only if, $p = 2$. 

Assume that $\vy = \mB \vz$. Specializing to rows $2$ to $m+1$, 
we obtain $[\vec{1}_{m-1}^{\T},-1]^{\T} = \mB'\vz$ with $\mB'$ the restriction of $\mB$ to the considered $m$ rows. We let the reader check that $\mB'$ is invertible and that this implies $\vz = \ve_{1}$. 
Specializing now to the first row, we have $2^{p-1} = y_{1} = (\mB \ve_{1})_{1} = 2$ that is to say $p=2$. The converse is immediate.
This shows the desired property with $\mA_{2} \bydef \mB^{\dag}$.

To further characterize the role of $q$ when $p=2$ we now consider $m \geq 3$ and construct a full rank $\mB$ from $\mC$ where $\mC$ is binary with the same number of ones in all rows except the first one, so that $\mD_{\text{row}}$ has all diagonal entries equal except the first one and $\mD_{\text{row}}^{q-p} = c (\mI + \lambda \ve_{1}\ve_{1}^{\T})$ where $c>0$, and $\lambda \neq 0$ as soon as $q \neq p$. With $\vb^{1}$ the first row of $\mB$ (a row vector),
\begin{eqnarray*}
    \mD_{\text{row}}^{q-p}\abs{\mB}^{p-1} 
    &=&
    c \abs{\mB}^{p-1} + c\lambda \ve_{1} \abs{\vb^{1}}^{p-1}\\
    &\stackrel{\eqref{eq:Bconstruction}}{=}& 
    c(2^{p-1}-1) \mB - c (2^{p-1}-2) \vec{1}_{n}\vec{1}_{m}^{\T} + c\lambda \ve_{1} \abs{\vb^{1}}^{p-1}\\
    (\mI_{n}-\mP_{\mB}) \mD_{\text{row}}^{q-p}\abs{\mB}^{p-1} 
    &=&
    c(\mI_{n}-\mP_{\mB})\left(\lambda \ve_{1}\abs{\vb^{1}}^{p-1} - (2^{p-1}-2)\vec{1}_{n}\vec{1}_{m}^{\T}\right)\\
    &=& c\left( \lambda \vu \abs{\vb^1}^{p-1} - (2^{p-1}-2) \vv \vec{1}_{m}^{\T}\right)
\end{eqnarray*}
with $\vu \bydef (\mI_{n}-\mP_{\mB})\ve_{1}$ and $\vv \bydef (\mI_{n}-\mP_{\mB})\vec{1}_{n}$.

For $m \geq 3$, consider the specific choice 
\[
    \mC = 
    \left[\begin{array}{c}
   \vc^{1}\\
        \hline
        \mI_{m}\\ 
        \hline
        \mat{0}_{(n-1-m) \times (m-1)}, \ \mat{1}_{n-1-m}
    \end{array}\right]
\]
where the last block is empty if $n=m+1$, and the first row is 
$\vc^{1} =  [1,1,\mat{0}_{m-2}^{\T}]$.

For any vector $\vz$ and scalars $\alpha,\beta$ such that $\mB\vz = \alpha \vec{1}_{n} + \beta\ve_{1}$, one can show that 
$\vz = \beta \vec{1}_{m}$ and $\alpha = (m+1)\beta$.
In particular: 
\begin{enumerate}[label=(\roman*)]
    \item The matrix $\mB$ is full rank as the equality  $\mB\vz = \vec{0}_{n} = \alpha\vec{1}_{n}+\beta\ve_{1}$ with $\alpha=\beta=0$ implies $\vz = \vec{0}_{m}$. 
    \item Similarly $\mB \vz = \alpha \ve_1 + \beta\vec{1}_{n}$ with $\beta=0$ implies $\vz = \vec{0}$, hence $\ve_1 \notin\range(\mB)$ and  $\vu \neq \vec{0}$;
    \item By the same argument, $\vec{1}_{n} \notin \range(\mB)$, and $\vv \neq \vec{0}$;
    \item For $m\geq 3$, $(m+1) \vec{1}_n + \vec{e}_1 = \mB \vec{1}_m$ \emph{does} belong to the span of $\mB$. 
    Up to a scaling it is the only linear combination of $\vec{1}_n$ and $\vec{e}_1$ that lives there; if there were another one, then by linearity the whole 2D subspace would be in the span of $\mB$ including $\vec{1}_n$ and $\ve_1$ for which we have shown that this is not the case;
    \item This implies that $\vu$ and $\vv$ are colinear with $\vu = -(m+1)\vv$
    as $\alpha \vv + \beta \vu = \vec{0}$ if, and only if, $\alpha \vec{1}_{n}+\beta \ve_{1} \in \range(\mB)$, which is equivalent to $\alpha = (m+1)\beta$.
\end{enumerate}
As a result
\[
    (\mI_{n}-\mP_{\mB}) \mD_{\text{row}}^{q-p}\abs{\mB}^{p-1} 
    = c\cdot\vu\left(\lambda \abs{\vb}^{p-1}  + 
    \tfrac{2^{p-1}-2}{m+1} \vec{1}_{m}^{\T}\right).
\]
This is zero if and only if 
\begin{equation}\label{eq:IntermediateGinvRow}
\lambda \abs{\vb}^{p-1}  = 
\tfrac{2^{p-1}-2}{m+1} \vec{1}_{m}^{\T}.
\end{equation}


As $\abs{\vb^{1}}^{p-1} = [2^{p-1},2^{p-1},\vec{1}_{m-2}^{\T}]$, ~\eqref{eq:IntermediateGinvRow} is equivalent to
\begin{eqnarray}
    \lambda &=& -(m+1)^{-1}(2^{p-1}-2)2^{-(p - 1)}\label{eq:IntermediateGinvRow1}\\
    \lambda  &=& -(m+1)^{-1}(2^{p-1}-2).\label{eq:IntermediateGinvRow2}
\end{eqnarray}

\begin{itemize}
\item 
For $p \notin \set{1,2}$ and any $q$, the right-hand side of these two equalities have incompatible values, so the only way for~\eqref{eq:IntermediateGinvRow1}-\eqref{eq:IntermediateGinvRow2} to hold is to have $p \in \set{1,2}$. 

\item 
For $p=2$ the right-hand side of these equalities is zero, hence they are satisfied only if $\lambda = 0$. Since $\mD_{\text{row}}^{p} = \mD_{\text{row}}^{2} = \diag(m+6,m+3,\ldots,m+3) = (m+3) \diag(\tfrac{m+6}{m+3},1,\ldots,1)$, we have $\lambda = \left(\tfrac{m+6}{m+3}\right)^{q/p-1}-1$ and the condition $\lambda = 0$ can only hold if $q=p$. This means that the only way for~\eqref{eq:IntermediateGinvRow} to hold when $p=2$ is to have $(p,q)=(2,2)$.

\item 
For $p=1$ and any $q$, the right-hand sides are compatible and yield the constraint $\lambda+1=\tfrac{m+2}{m+1}$. For this value of $p$ we have $\mD_{\text{row}} = \diag(m+2,m+1,\ldots,m+1)$, $c = (m+1)^{q-1}$, and $\lambda = \left(\tfrac{m+2}{m+1}\right)^{q-1}-1$. Combining we get the condition
\begin{equation}
    \label{eq:cond_rowwise_notmpp}
    \frac{m+2}{m+1} = \left( \frac{m+2}{m+1} \right)^{q-1}
\end{equation}
which can only hold for $q = 2$. 
\end{itemize}
This shows the desired property with $\mA_{3} \bydef \mB^{\dag}$.

For $(p,q) = (1,2)$, the construction of $\mA_{3}$ does not give us a counterexample. In fact, for this choice of $p$ and $q$, whenever $\mB$ has all entries positive, we have $\mD_{\text{row}}^{q-p} = \mD_{\text{row}} = \diag(\mB \vec{1}_{m})$ and
\begin{eqnarray*}
    (\mI_n - \mP_\mB) \mD^{q-p}_\text{row} \left[ \abs{\mB}^{p-1} \odot \sign(\mB) \right]
    & = & (\mI_n - \mP_\mB) \diag\left(\mB \vec{1}_{m}\right) \vec{1}_n\vec{1}_m^\T \\
    & = & (\mI_n - \mP_\mB) \mB \vec{1}_{m}\vec{1}_m^\T = \mat{0}_{n \times m}
\end{eqnarray*}
where the last equality follows because $
\mB \vec{1}_m$ is always in the range of $\mB$. However the counter-example built with $\mA_{2}$ is valid for $(p,q)=(1,2)$ hence $(p,q) \neq (2,2)$ implies the existence of a full-rank $\mA$ such that $\mA^{\dag} \notin \ginv{\rowpq{p}{q}}{\mA}$.

\item {\bf Case of $\pginv{}{\cdot}$ for columnwise norms} 

For columnwise norms, by Lemma~\ref{le:subgradients0}, we seek $\mA$ so that
\[
h^{\mathsf{pginv}}_{\colpq{p}{q}}(\mA) 
=  (\mI_{n}-\mA^{\dag}\mA) \left( \abs{\mA^{\dag}\mA}^{p-1} \odot \sign(\mA^{\dag}\mA)\right) \diag\left(\norm{(\mA^{\dag}\mA)_{j}}_{p}^{q-p}\right)_{j=1}^{n}\mA^{\T} \ni \mat{0}
\]
only for specific values of $(p,q)$, where for $p=1$ the notation $\abs{\mM}^{p-1}\odot \sign(\mM)$ should be replaced by $\sign(\mM)$ 
Again, some manipulations indicate that it is sufficient to find a full rank $\mB \in \R^{n \times m}$ so that 
$\mP_{\mB}$ has no zero column and we have, with $\mD = \diag(\norm{(\mP_{\mB})^{i}}_{p})$
\begin{equation}\label{eq:strategyPGinvCol}
(\mI_{n}-\mP_{\mB}) (\abs{\mP_{\mB}}^{p-1} \odot \sign(\mP_{\mB})) \mD^{q-p}
\mP_{B}
\ni \mat{0}
\end{equation}
only for these values of $(p,q)$. In other words, it is sufficient to find an orthogonal projection matrix $\mP$ of rank $m$ with no zero column so that 
with $\mD = \diag(\norm{(\mP)^{i}}_{p})$:
\[
(\mI_{n}-\mP) (\abs{\mP}^{p-1} \odot \sign(\mP)) \mD^{q-p} \mP \ni \mat{0}
\]
only for controlled values of $(p,q)$.

For rowwise norms, by Lemma~\ref{le:subgradients0}, we need to exhibit $\mA$ so that we control for which $(p,q)$ we have
\[
h^{\mathsf{pginv}}_{\rowpq{p}{q}}(\mA) 
=  (\mI_{n}-\mA^{\dag}\mA) \diag\left(\norm{(\mA^{\dag}\mA)^{i}}_{p}^{q-p}\right)_{j=1}^{n} \left( \abs{\mA^{\dag}\mA}^{p-1} \odot \sign(\mA^{\dag}\mA)\right) \mA^{\T}
\ni \mat{0}.
\]
With the same reasoning as above it is sufficient to find an orthogonal projection matrix $\mP \in \R^{n \times n}$ of rank $m$ with no zero column so that  we control the values of $(p,q)$ for which 
\begin{equation}\label{eq:strategyPGinvRow}
(\mI_{n}-\mP) \mD^{q-p} (\abs{\mP}^{p-1} \odot \sign(\mP)) 
\mP
\ni \mat{0}
\end{equation}
where $\mD =  \diag(\norm{(\mP)_{j}}_{p}) = \diag(\norm{(\mP)^{i}}_{p})$ by the symmetry of $\mP$.

We now proceed to the desired constructions for various dimensions $m,n$.
\begin{itemize}
\item For $m=1$ we have for any $0<p,q \leq \infty$ and $\mX \in \R^{n \times 1}$: $\normcolpq{p}{q}{\mX\mA} =  \norm{\mX}_{p}\norm{\mA}_{q}$ and $\normrowpq{p}{q}{\mX\mA} = \norm{\mX}_{q}\norm{\mA}_{p}$. Hence for any $0<p,q\leq\infty$ and any $\mA \in \R^{1 \times n}$, $\pginv{\colpq{p}{q}}{\mA} = \ginv{p}{\mA}$ and $\pginv{\rowpq{p}{q}}{\mA} = \ginv{q}{\mA}$. This allows reusing the matrix $\mA_{1}$.

\item For $m \geq 2$ and columnwise norms, we construct $\mP$ as follows:
choose $0< \theta < \pi/2$, $\theta \neq \pi/4$ and set $c \bydef \cos \theta > 0$, $s \bydef \sin\theta>0$, $\vu \bydef (c,s)$; build a block diagonal projection matrix of rank $m$, 
\[
\mP \bydef \blockdiag(\vu\vu^{\T},\mP')
\]
where $\mP'$ is an arbitrary $(n-2) \times (n-2)$ projection matrix of rank $m-1$ with nonzero columns.
Denoting
 $\vv \bydef (-s,c)$ and $\vw \bydef (c^{p-1},s^{p-1})$, we have $\mI_{2}-\vu\vu^{\T}= \vv\vv^{\T}$ and
\begin{eqnarray*}
\mI_{n}-\mP &=& \blockdiag(\vv\vv^{\T},\mI_{n-m}-\mP')\\
\mD &=& \blockdiag(\|\vu\|_{p}\diag(\vu),\mD')\\
\abs{\mP}^{p-1} \odot\sign(\mP) &=& \left\{
\begin{array}{ll}
\blockdiag(\vw\vw^{\T},\mW), & \text{for}\ 1< p < \infty\\
\sign(\mP) = \set{\begin{bmatrix}
\vw\vw^{\T} & \mS_1\\
\mS_{2} & \mW
\end{bmatrix}, 
\normcolpq{\infty}{\infty}{\mS_{i}} \leq 1
} & \text{for}\ p=1,
\end{array}
\right.
\end{eqnarray*}
where $\mD'$ is diagonal and $\mW$ is some matrix. For $p = 1$ we abused the notation and instead of specifying the whole subdifferential as a set, we kept the decomposition into parts and assigned the set to the only ambiguous term. For $p>1$ it follows that
\begin{eqnarray*}
(\mI_{n}-\mP) (\abs{\mP}^{p-1}\odot\sign(\mP)) \mD^{q-p} \mP
&=&
\|\vu\|_{p}^{q-p}\ 
\blockdiag(\vv\vv^{\T}\vw\vw^{\T}\diag(\abs{\vu}^{q-p})\vu\vu^{\T},\mW')\\
&=&
\|\vu\|_{p}^{q-p}\cdot\langle\vv,\vw\rangle \cdot (\vw^{\T}\diag(\abs{\vu}^{q-p})\vu)\ \\
& & \cdot\ \blockdiag(\vv\vu^{\T},\mW'').
\end{eqnarray*}
for some matrices $\mW'$, $\mW''$.
For $p=1$ we similarly get
\begin{eqnarray*}
(\mI_{n}-\mP) \sign(\mP) \mD^{q-p} \mP
&=&
\|\vu\|_{p}^{q-p}
\cdot \langle\vv,\vw\rangle 
\cdot \vw^{\T}\diag(\abs{\vu}^{q-p})\vu
\set{\begin{bmatrix}
\vv\vu^{\T} & \mW_{1}\\
\mW_{2} & \mW_{3}
\end{bmatrix},
\mW_{i} \in \mathcal{W}_{i}
}
\end{eqnarray*}
for some matrix sets $\mathcal{W}_{i}$.

As $\vv\vu^{\T} \neq \mat{0}$, $\vw^{\T}\diag(\abs{\vu}^{q-p})\vu = c^{q} + s^{q} > 0$, and $\langle\vv,\vw\rangle = -sc^{p-1}+cs^{p-1}$ this is zero if, and only if, $c s^{p-1} = s c^{p-1}$, i.e. $(s/c)^{p-2} = 1$. As $s/c=\tan\theta \neq 1$, this is equivalent to $p = 2$. 

\end{itemize}

\item {\bf Case of $\pginv{}{\cdot}$ for rowwise norms} 

\begin{itemize}

\item 

For $m \geq 2$, we proceed similarly as for the columnwise norms:
\begin{eqnarray*}
(\mI_{n}-\mP) \mD^{q-p} (\abs{\mP}^{p-1}\odot\sign(\mP))\mP
&=&
\|\vu\|_{p}^{q-p}\ 
\blockdiag(\vv\vv^{\T}\diag(\abs{\vu}^{q-p})\vw\vw^{\T}\vu\vu^{\T},\mW)\\
&=&
\|\vu\|_{p}^{q-p}\cdot\vv^{\T}\diag(\abs{\vu}^{q-p})\vw \cdot\ \langle\vw,\vu\rangle\ \\
& & \cdot\ \blockdiag(\vv\vu^{\T},\mW').
\end{eqnarray*}
with the appropriate adaptation for $p=1$. As $\langle \vw,\vu\rangle = c^{p}+s^{p}>0$ and $\vv^{\T}\diag(\abs{\vu}^{q-p})\vw = -s c^{q-1}+c s^{q-1}$, the same reasoning shows this is zero if, and only if, $q = 2$.

\item For $m=2$, $n=3$, $q=2$, consider the unit norm vector $\vv \bydef [2,1, 1]^{\T}/\sqrt{6}$ and define $\mP \bydef \mI_{3}-\vv\vv^{\T}$, which is a rank-$(m=2)$ projection matrix. We have for any $0<p<\infty$
\begin{eqnarray*}
\mP  =  \tfrac{1}{6} 
\begin{bmatrix}
2 & -2 & -2\\
-2 & 5 & -1\\
-2 & -1 & 5
\end{bmatrix} 
&=& \tfrac{1}{3}
\begin{bmatrix}
1 & -1 & -1\\
-1 & 5/2 & -1/2\\
-1 & -1/2 & 5/2
\end{bmatrix}
\\
\abs{\mP}^{p-1} \odot \sign(\mP) & \propto & 
\begin{bmatrix}
1 & -1 & -1\\
-1 & (5/2)^{p-1} & -(1/2)^{p-1}\\
-1 & -(1/2)^{p-1} & (5/2)^{p-1}
\end{bmatrix} \bydef \mQ_{p}
\end{eqnarray*}
Moreover denoting $c_{p} \bydef 1+(5/2)^{p}+(1/2)^{p}$ we have $\mD_{\text{row}}^{q-p} \propto (\mI_{3}+\lambda \ve_{1}\ve_{1}^{\T})$ with $\lambda = (3/c_{p})^{q/p-1}-1$, hence 
\begin{eqnarray}
(\mI_{3}-\mP)  \mD^{q-p} (\abs{\mP}^{p-1} \odot \sign(\mP))\mP
& \propto & \vv\vv^{\T} \left(\mI_{3}+\lambda \ve_{1}\ve_{1}^{\T}\right) \mQ_{p}\mQ_{2}
\label{eq:PGinvRowM2}
\end{eqnarray}
As $\vv \neq \vec{0}_{3}$, the right-hand side of~\eqref{eq:PGinvRowM2} vanishes if and only if
\begin{eqnarray}
\vec{0}_{3} &=& \mQ_{2} \mQ_{p} (\mI_{3}+\lambda \ve_{1}\ve_{1}^{\T}) \vv 
=\mQ_{2}\mQ_{p}\vv + \lambda v_{1} \mQ_{2}\mQ_{p}\ve_{1}
\propto \mQ_{2}\mQ_{p} \begin{bmatrix}2\\1\\1\end{bmatrix} + 6\lambda \begin{bmatrix}1\\-1\\-1\end{bmatrix}.\label{eq:PGinvRowM2N3} 
\end{eqnarray}
We let the reader check that 
\[
\mQ_{2}\mQ_{p} \begin{bmatrix}2\\1\\1\end{bmatrix}
= 2(2-(5/2)^{p-1}+(1/2)^{p-1})\begin{bmatrix}1\\-1\\-1\end{bmatrix}
\]
hence~\eqref{eq:PGinvRowM2N3} is equivalent to $\lambda = -((1/2)^{p-1}-(5/2)^{p-1}+2)/3$. We derived before that $\lambda = (3/c_{p})^{q/p-1}-1$ hence the right-hand side of~\eqref{eq:PGinvRowM2} vanishes if and only if:
\begin{equation}
    \label{eq:pginv_row_withlam}
    \left[ \frac{3}{1 + \left(\frac{5}{2}\right)^{p} + \left( \frac{1}{2}  \right)^p } \right]^{\frac{2-p}{p}} = \frac{1}{3} \left[ -\left( \frac{1}{2}\right)^{p-1}   + \left( \frac{5}{2} \right)^{p-1}  +1 \right].
\end{equation}
Setting $g(p) \bydef \frac{1}{3} \left[1 + \left(\frac{5}{2}\right)^{p} + \left( \frac{1}{2}  \right)^p \right]$ for $p \in \R$ this can be written for $0 < p < \infty$ as
\[
    g(p)^{\frac{p-2}{p}} = g(p - 1) - \frac{1}{3}\left( \frac{1}{2} \right)^{p-2}.
\]

\begin{itemize}

\item

For $p = 2$ the equation is satisfied and the right-hand side of \eqref{eq:PGinvRowM2} indeed vanishes. 

\item

For $p > 2$, note that $g(p) < \frac{1}{2} \left( \frac{5}{2} \right)^p$ since it is equivalent to
\[
    \tfrac{1}{6} \left( \tfrac{5}{2} \right)^p > \tfrac{1}{3} + \tfrac{1}{3} \left( \tfrac{1}{2} \right)^p
\]
where the left-hand side is strictly increasing, and the right-hand side strictly decreasing and smaller than the left-hand side for $p = 2$. Therefore, for $p>2$, 
\[
    g(p)^{\tfrac{p-2}{p}} < \left[\tfrac{1}{2} \left( \tfrac{5}{2} \right)^p\right]^{\tfrac{p-2}{p}} = \tfrac{2^{2/p}}{5} \left( \tfrac{5}{2} \right)^{p-1}.
\]
On the other hand, 
\[
    \left[\tfrac{1}{2} \left( \tfrac{5}{2} \right)^p\right]^{\tfrac{p-2}{p}}  <  g(p-1)-\tfrac{1}{3} \left(\tfrac{1}{2}\right)^{p-2}.
\]
To see that this is true, we rewrite it as
\[
    \tfrac{3\cdot 2^{2/p} - 5}{15} \left( \tfrac{5}{2} \right)^{p-1}  < \tfrac{1}{3} - \tfrac{1}{3} \left( \tfrac{1}{2} \right)^{p-1}
\]
which holds since for $p = 2$ the two sides coincide, the right-hand side is strictly increasing, and the left-hand side strictly decreasing (the latter can be verified directly by studying the sign of the derivative).%
\footnote{The derivative of the left-hand side is $\left\{\left(\frac{2}{5}\right)^{1-p} \left[4^{1/p} \left(3 p^2 \log \left(\frac{5}{2}\right)-2 \log (8)\right)-5 p^2 \log \left(\frac{5}{2}\right)\right]\right\}/(15 p^2)$.
When $p > \log(4) / \log(5/3) \approx 2.71$, then $4^{1/p} \cdot 3 < 5$ and the expression is negative. For $ 2 < p < 3$,
$4^{1/p} \left(3 p^2 \log \left(\frac{5}{2}\right)-2 \log (8)\right) -5 p^2 \log \left(\frac{5}{2}\right) < 2\left(3 p^2 \log \left(\frac{5}{2}\right)-2 \log (8)\right) -5 p^2 \log \left(\frac{5}{2}\right) = p^2 \log \left(\frac{5}{2}\right) - 4 \log(8) < 9 \log \left(\frac{5}{2}\right) - 4 \log(8) \approx -0.07$.
}

In summary, for $p > 2$ we have
\[
    g(p)^{\tfrac{p-2}{p}} < \tfrac{2^{2/p}}{5} \left( \tfrac{5}{2} \right)^{p-1} < g(p - 1) - \left( \tfrac{1}{2} \right)^{p-2}
\]
and the equation \eqref{eq:pginv_row_withlam} cannot be consistent.

\item

For $0 < p < 2$, the following holds:
\[
    g(p)^{(p-2)/p} > 2^{p-2}.
\]
To see this, note that it is equivalent to (remember that $p - 2 < 0$ so raising both sides to the power $1/(p - 2)$ reverses the inequality)
\[
    1 + \left( \tfrac{5}{2} \right)^p + \left( \tfrac{1}{2} \right)^p <  3\cdot 2^p,
\]
or
\[
    \left( \tfrac{1}{2} \right)^p + \left( \tfrac{5}{4} \right)^p + \left( \tfrac{1}{4} \right)^p < 3,
\]
which holds since $p \mapsto \left( \tfrac{1}{2} \right)^p + \left( \tfrac{5}{4} \right)^p + \left( \tfrac{1}{4} \right)^p$ is convex on $\R$, equal to 3 for $p = 0$, and equal to $\frac{15}{8} < 3$ for $p = 2$. On the other hand, we have that
\[
    2^{p-2} > g(p-1) - \tfrac{1}{3}\left(\tfrac{1}{2}\right)^{p-2}
\]
for $0 < p < 2$ and that the equation holds for $p = 0$. This follows since the two sides coincide for $p = 2$ and since the function
\[
    a(p) \bydef 2^{p-2} - \left[ g(p-1) - \tfrac{1}{3} \left(\tfrac{1}{2}\right)^{p-2}\right]
\]
is decreasing on $[0,\ 2]$. Indeed
\[
    a'(p) = \tfrac{1}{15} 2^{-p-2} \underbrace{\left[4^p \cdot 15 \log(2) - 5^p \cdot 8\log \left( \tfrac{5}{2} \right) - 40 \log(2) \right]}_{\bydef b(p)},
\]
and
\[
    b'(p) = 15\cdot 4^p \log (2) \log (4)-8\cdot 5^p \log \left(\tfrac{5}{2}\right) \log (5).
\]
It can be verified that $b(p)$ has a single critical point on $[0, 2]$ at $p_1 \bydef \log \left( \tfrac{15 \log (2) \log (4)}{8 \log \left(\tfrac{5}{2}\right) \log (5)}\right) / \log \left( \tfrac{5}{4} \right)$ which is a maximum, and that $b(p_1) \approx -22.72$. Hence, $a'(p) < 0$ on $(0, 2)$.
\end{itemize}
In conclusion, for \eqref{eq:PGinvRowM2N3} to vanish when $q=2$, it must hold that $p=2$.

\item For any $3 \leq m < n$, one can find $\mP_{1} \in \R^{(n-3)\times(n-3)}$, a projection matrix of rank $m-2 \geq 1$ with no zero column, and build $\mP' = \blockdiag(\mP,\mP_{1})$ with the rank-2 matrix $\mP \in \R^{3 \times 3}$ we have just built. The same reasoning as before leads to a construction so that for $1 \leq p < \infty$, $q=2$, $\mA_{5}^{\dag} \in \pginv{\rowpq{p}{q}}{\mA_{5}} \Longleftrightarrow p=2$.
\end{itemize}     

\end{enumerate}
\end{proof}

\bibliographystyle{IEEEtranSA}

\bibliography{pseudo,pseudo_supp}

\end{document}